\documentclass[oneside,reqno]{amsart}
\usepackage{amssymb}
\newtheorem{theorem}{Theorem}
\newtheorem{lemma}[theorem]{Lemma}
\theoremstyle{remark}
\newtheorem*{remark}{Remark}
\newcommand{\myStrut}{\parbox{0.11 pt}{\rule{0 ex}{4.5 ex}}}
\newcommand{\myStrutB}{\parbox{0.11 pt}{\rule{0 ex}{8.5 ex}}}
\newcommand{\myStrutC}{\parbox{0.11 pt}{\rule{0 ex}{4.9 ex}}}
\newcommand{\be}{\begin{equation}}
\newcommand{\ee}{\end{equation}}
\newcommand{\bc}{\begin{center}}
\newcommand{\ec}{\end{center}}
\newcommand{\bmt}{\begin{pmatrix}}
\newcommand{\emt}{\end{pmatrix}}
\newcommand{\tbins}[2]{\parbox{#1 in}{\bc #2\ec}}
\newcommand{\urc}{\mathcal{C}_d}
\newcommand{\ure}{\mathcal{E}_d}
\newcommand{\rc}{\mathcal{C}^{\mathrm{r}}_d}
\newcommand{\re}{\mathcal{E}^{\mathrm{r}}_d}
\DeclareMathOperator{\ord}{ord}
\DeclareMathOperator{\tr}{tr}
\DeclareMathOperator{\Tr}{Tr}

\DeclareMathOperator{\SL}{SL}
\DeclareMathOperator{\ESL}{ESL}
\begin{document}
\begin{titlepage}
\begin{center}
\bfseries  PROPERTIES OF THE EXTENDED CLIFFORD GROUP WITH APPLICATIONS TO SIC-POVMS AND MUBS.
\end{center}
\vspace{1 cm}
\begin{center} D M APPLEBY
\end{center}
\begin{center} Perimeter Institute for Theoretical Physics \\ Waterloo, Ontario N2L 2Y5, Canada
 \end{center}
\vspace{0.5 cm}
\begin{center}
  (E-mail:  mappleby@perimeterinstitute.ca)
\end{center}
\vspace{0.75 cm}
\vspace{1.25 cm}
\begin{center}
\vspace{0.35 cm}
\parbox{12 cm }{We consider a version of the extended Clifford Group  which is defined in terms of a finite Galois field in odd prime power dimension.  We show that Neuhauser's result, that with the appropriate choice of phases the standard (or metaplectic) representation of the discrete symplectic group is faithful (as opposed to merely projective), also holds for the anti-unitary operators of the extended group.  We also improve on Neuhauser's result by giving explicit formulae which enable one easily to calculate the  (anti-)unitary corresponding to an arbitrary (anti-)symplectic matrix.  We then go on to find the eigenvalues of an arbitrary (anti-)symplectic matrix.  The fact that in prime power dimension the matrix elements belong to a field means that this can be done using  the same techniques which are used to find the eigenvalues of a matrix defined over the reals---including the use of an extension field (the analogue of the complex numbers) when the eigenvalues are not in the base field.  With the eigenvalues of the (anti-)symplectic matrix in hand it is straightforward to find the eigenvalues, the order and all roots of the corresponding (anti-)unitary.   We then give an application of these results to SIC-POVMs (symmetric informationally complete positive operator valued measures).  We show that in prime dimension our results can be used to find a natural basis for the eigenspace of the Zauner unitary in which SIC-fiducials are expected to lie. We go on to use the extension field to construct a parameterization of the displacement operators analagous to the one commonly employed in quantum optics, using the complex eigenvalues of the annihilation operator.   Finally, we apply our results to the MUB cycling problem.  Wootters and Sussman have shown that in every even prime power dimension there is a single Clifford unitary which cycles through a full set of MUBs (mutually unbiassed bases).  We show that in odd prime power dimension $d$, although there is no Clifford unitary, there is a Clifford anti-unitary which cycles through the full set of  Wootters-Fields MUBs if $d=3$ (mod $4$).  Also, irrespective of whether $d=1$ or $3$ (mod $4$), the Wootters-Fields MUBs split into two groups of $(d+1)/2$ bases in such a way that there is a single Clifford unitary which cycles through each group separately.}
\end{center}
\end{titlepage}

{\allowdisplaybreaks
\section{Introduction}
Since its original introduction into the field of quantum information the Clifford group~~\cite{Gottesman1, Gottesman2,Gottesman3, Neuhauser, VourdasErrB, VourdasA, Dehaene,vanDenNestD, VourdasErrA,VourdasB,WoottersC,selfB,vanDenNest,Hostens,vanDenNestC, vanDenNestB, GrossB, Galvao,Paz, VourdasC,CormickEtAl, Gross, Flammia, VourdasD, WoottersA, Sussman, GrossC, selfC, elliot, KiblerF} has found numerous applications.  Our own interest in the group, and the extended Clifford Group~\cite{selfB,selfC} of which it is a subgroup, is due to their role in the theory of mutually unbiassed bases (MUBs)~\cite{WoottersC,Galvao,Paz,CormickEtAl,WoottersA,Sussman,selfC,Ivanovic,WoottersD,WoottersB,Bandy,Lawrence,Chaturvedi,Archer,PittA,KlapRot,DurtB,GrasslA,KlimovA,KlimovB,Planat,KlappA,Ingemar2,Durt,PittB,Wocjan,GodsilRoy,WoottersE,KiblerA,KiblerB,KlimovC,Scott,KlimovD,MubsD6,  BrierleyA,KiblerE,selfE,KiblerC,BrierleyB,BrierleyC} and symmetric, informationally complete positive operator valued measures (SIC-POVMs)~\cite{selfB,GrossB,Flammia,GrasslA,KlappA,GodsilRoy,WoottersE,Scott,KiblerE,selfE,Hoggar,Zauner,Caves,FuchsA,  Renes,  SanigaPlanatRosu,FuchsB, Rehacek, KlappB,GrasslB,Colin, Kim,me2,selfA,Khat,Ingemar3,ScottGrassl}.  However, the results we derive along the way may be of  wider interest. 

We should remark at the outset that it is possible to identify at least three different (though closely related) mathematical constructs which commonly go by the name ``Clifford group''.  We will clarify the precise sense in which we are using the term  in Section~\ref{sec:prelim} (see paragraph following Eq.~(\ref{eq:dLinDef})).  For now we will confine ourselves to saying that we are concerned with a version  of the Clifford group which is defined over a finite field~\cite{LidNie}, and which therefore  only exists in prime power dimension.  There are some significant complications when the dimension is a power of $2$, so in this paper we confine ourselves to the case of odd prime power dimension.  We hope to treat the even prime power case in a subsequent paper.

In the first part of the paper (sections~\ref{sec:prelim} to~\ref{sec:orders}) we prove a number of general results concerning the  Clifford and extended Clifford groups.    In the second part we apply these results to a problem which arises in the theory of SIC-POVMs and to the MUB cycling problem. We also show how they can be used to construct an alternative parameterization of the displacement operators.    In a subsequent paper we will use them to complete the proof that the minimum uncertainty states introduced by Wootters and Sussman~\cite{WoottersA,Sussman} and Appleby, Dang and Fuchs~\cite{selfA} exist in every prime power dimension.

The Clifford group is built out of two kinds of operator:  displacement operators $D_{\mathbf{u}}$, labelled by   vectors of the form
\be
\mathbf{u} = \bmt u_1 \\ u_2 \emt
\ee
with components in the finite field $\mathbb{F}_d$ ($d$ being the Hilbert space dimension), and symplectic unitaries $U_F$, labelled by  $2\times 2$ matrices of the form
\be
F = \bmt \alpha & \beta \\ \gamma & \delta \emt
\ee
with elements in $\mathbb{F}_d$ and determinant $=1$.  The set of all such matrices constitutes the (discrete) symplectic group, $\SL(2,\mathbb{F}_d)$.  The symplectic unitaries permute the displacement operators according to the prescription
\be
U^{\vphantom{\dagger}}_F D^{\vphantom{\dagger}}_{\mathbf{u}} U^{\dagger}_F =  D^{\vphantom{\dagger}}_{F\mathbf{u}}
\ee
Explicit expressions for the operators $D_{\mathbf{u}}$ and $U_F$ will be given below.  

The extended Clifford group is built out of the above operators together with anti-symplectic anti-unitaries $U_F$, for which the matrix $F$ has determinant $=-1$. We refer to the set of all $2\times 2$ matrices with either sign of the determinant  as the extended symplectic group $\ESL(2,\mathbb{F}_d)$.   

After reviewing what is already known on this subject in Section~\ref{sec:prelim} we begin our analysis in Section~\ref{sec:Faithful} by showing that the phases of the operators  $U_F$ can be chosen in such a way that the map $F\to U_F$ (the metaplectic representation~\cite{Neuhauser,Gross}) becomes a faithful representation of the group $\ESL(2,\mathbb{F}_d)$, as opposed to one that is merely projective.  In other words
\be
U_{F_1} U_{F_2} = U_{F_1F_2}
\ee
for all $F_1$, $F_2$.  The fact that this is possible for the group $\SL(2,\mathbb{F}_d)$ was shown by Neuhauser~\cite{Neuhauser}.  We establish that the result extends to  $\ESL(2,\mathbb{F}_d)$.  We also improve on Neuhauser's result by giving explicit formulae which enable one easily to calculate $U_F$ for any given matrix $F$.  In fact we give two different sets of formulae for, in addition to giving explicit expressions for the matrix elements, we show (in Section~\ref{sec:UFtermsDu}) that $U_F$ can be simply expressed as a linear combination of displacement operators.

We next (Section~\ref{sec:orders})  consider the problem of finding the order of an arbitrary symplectic/anti-symplectic operator $U_F$ (\emph{i.e.} the smallest positive integer $m$ such that $U^m_F =1$),  its eigenvalues and eigenspaces, and its roots. To solve this problem we make essential use of the fact that $\mathbb{F}_d$ is a field.  This means that we can use exactly the same techniques to diagonalize a matrix

\be
F=\bmt \alpha & \beta \\ \gamma & \delta \emt
\ee
with elements $\in \mathbb{F}_d$ that we would if its elements were real numbers.   In particular we can use the trick of going to an extension field.  If it should happen that the characteristic equation of a matrix over the reals has no real solutions, then we simply embed the reals in the field of complex numbers.  Similarly here:  if the matrix $F$ has no eigenvalues in $\mathbb{F}_{d}$ we can  embed $\mathbb{F}_d$ in the extension field $\mathbb{F}_{d^2}$.  Once one knows the eigenvalues, finding the order is a comparatively simple matter.  We give explicit formulae which enable one easily to calculate the order of an arbitrary matrix $F\in\ESL(2,\mathbb{F}_d)$.   With these formulae in hand it is  straightforward to calculate the order, eigenspaces and eigenvalues, and  roots of the corresponding unitary/anti-unitary $U_F$.  

In Section~\ref{sec:SICs} we apply the methods developed in the previous sections to a problem which arises in theory of SIC-POVMs.  Scott and Grassl's~\cite{ScottGrassl} recent exhaustive numerical investigation encourages the speculation that, in all dimensions, every Weyl-Heisenberg covariant SIC-fiducial is an eigenvector of a canonical order $3$ Clifford unitary~\cite{selfB}. Generally speaking this unitary belongs to a different version of the Clifford group than the one considered here.  However, if the dimension is prime (as opposed to a power $>1$ of a prime) the two versions are the same, and the results obtained in this paper become applicable.  We show how they can be used to find a natural basis for the eigenspace of the order $3$ symplectic unitary in which SIC-fiducials are expected to lie.

In Section~\ref{sec:FsAsPerms} we describe an alternative labelling of the displacement operators.  In quantum optics it is usual to parameterise the (continuous variable) displacement operators with the complex variable
\be
\alpha = \frac{1}{\sqrt{2}} (q + i p)
\ee 
(where $q$, $p$ are the quadratures).  A similar construction can be carried through in the discrete case, using the extension field $\mathbb{F}_{d^2}$ instead of $\mathbb{C}$. Taking the discrete logarithm we obtain an integer labelling of the displacement operators, which leads to a natural way of representing $\ESL(2,\mathbb{F}_d)$ as a group of permutation matrices.   This result is needed in Section~\ref{sec:reLabelling2}, but it may be of some independent interest.

Finally, in Sections~\ref{sec:MUBPrelim} to~\ref{sec:reLabelling2} we consider the MUB-cycling problem.  Wootters and Sussman~\cite{WoottersA,Sussman} have shown that in every dimension equal to a power of $2$ there exists a single Clifford unitary which cycles through a full set of  MUBs.  We investigate the situation in odd prime power dimension.  We show that in this case there is no single Clifford unitary which cycles through the full set of Wootters-Fields  MUBs~\cite{WoottersB} (at least for the version of the Clifford group considered here---see below).  However, if $d=3$ (mod $4$) there is a single \emph{anti}-Clifford \emph{anti}-unitary which cycles through them.  Furthermore, it is possible, for all odd prime power $d$,  to split the Wootters-Fields MUBs into two groups of $(d+1)/2$ bases each in such a way that there is a single Clifford unitary which cycles through each group separately (we say that a unitary with this property is half-cycling).  This leads to  a natural labelling scheme for the MUBs, in which a cycling anti-unitary (when it exists) increases the integer index by $1$ and a half-cycling unitary increases it by $2$.  

Wootters and Sussman~\cite{WoottersA} used the existence of cycling unitaries in even prime power dimension to prove the existence of  minimum uncertainty states  in all such dimensions.  Sussman~\cite{Sussman} subsequently extended the proof to an infinite subset of the prime power dimensions $=3$ (mod $4$).  In a subsequent paper we will use the results obtained in Sections~\ref{sec:MUBPrelim} to~\ref{sec:reLabelling2} of this paper to complete the proof, and to show that minimum uncertainty states exist in all prime power dimensions, without exception. 

\section{Preliminaries}
\label{sec:prelim}
We begin by reviewing some relevant definitions and known facts concerning Galoisian variants of the Clifford and extended Clifford groups~\cite{Neuhauser,VourdasErrB,VourdasA,VourdasErrA,VourdasB,WoottersC,GrossB,Paz,VourdasC,Gross,VourdasD,WoottersA,Sussman,selfC,WoottersB,PittA,KlapRot,KlimovA,Planat,Durt,PittB,GodsilRoy,KlimovC,KlimovD}.  This will also allow us to clarify the relation between the version of the Clifford group considered in this paper and  definitions used by other authors.

Let us begin by defining the Galoisian variant of the Weyl-Heisenberg group (or generalized Pauli group).  The ordinary Weyl-Heisenberg group is defined by choosing an orthonormal basis $|0\rangle, |1\rangle, \dots , |d-1\rangle$ in $d$-dimensional Hilbert space  and then defining operators $X$ and $Z$ by
\begin{align}
X | x\rangle & = | x+1 \rangle
\\
Z | x \rangle & = \omega^{x} | x \rangle 
\end{align}
where $\omega = e^{\frac{2 \pi i}{d}}$ and addition of indices is mod $d$.  If $d$ is odd the Weyl-Heisenberg group is the group generated by these operators (in even dimension there is a slight complication---see, for example, ref.~\cite{selfB}).  In dimension $d=p^n$, where $p$ is a prime number and $n$ is a positive integer, the Galoisian variant of the Weyl-Heisenberg group is defined similarly except that instead of labelling the orthonormal basis by the integers mod $d$ one labels them by the elements of the finite field $\mathbb{F}_d$ (for a summary account of the aspects of the theory of finite fields which are relevant here see, for example,  Vourdas~\cite{VourdasC}, Pittenger and Rubin~\cite{PittA} or Klimov \emph{et al}~\cite{KlimovC};   for a more comprehensive treatment see, for example, Lidl and Niederreiter~\cite{LidNie}).  One then defines, for all $x, u\in\mathbb{F}_d$,
\begin{align}
X_u | x\rangle & = | x+u \rangle
\label{eq:GalXopDef}
\\
Z_u | x \rangle & = \omega^{\tr(xu)} | x \rangle 
\label{eq:GalZopDef}
\end{align}
where $\omega=e^{\frac{2 \pi i}{p}}$ (observe that $\omega$ is now a $p^{\mathrm{th}}$ root of unity, not a $d^{\mathrm{th}}$ root of unity as in the case of the ordinary variant), and where  $\tr \colon \mathbb{F}_d\to \mathbb{Z}_p$ is the field theoretic trace defined by
\be
\tr (x) = \sum_{r=0}^{n-1} x^{(p^r)}
\ee
(making the natural identification of $\mathbb{Z}_p$, the integers modulo $p$, with the set $\{z\in \mathbb{F}_d \colon z^p = z\}$).  Notice that if $n>1$ we cannot write $X_u = X^u$, $Z_u = Z^u$ because it does, in general, make no sense to raise a Hilbert space operator to the power of an element of a finite field.  However if $n=1$ we can identify  $\mathbb{F}_d$  with $\mathbb{Z}_p$, which means we can make these replacements, and the Galoisian variant of the Weyl-Heisenberg group becomes identical with the ordinary variant.

Next, for each vector $\mathbf{u}=(u_1,u_2) \in \mathbb{F}^2_d$, define the displacement operator
\be
D_{\mathbf{u}} = \tau^{\tr(u_1 u_2)} X_{u_1} Z_{u_2}
\ee
where $\tau= \omega^{\frac{p+1}{2}}$.  Note that $\tau^2 = \omega$ and that $\tau$, like $\omega$, is a $p^{\mathrm{th}}$ root of unity.  
We have, as an immediate consequence of this definition,
\begin{align}
D^{\dagger}_{\mathbf{u}} &= D^{\vphantom{\dagger}}_{-\mathbf{u}}
\\
\intertext{and}
D_{\mathbf{u}} D_{\mathbf{v}} &= \tau^{\langle \mathbf{u},\mathbf{v}\rangle } D_{\mathbf{u}+\mathbf{v}}
\label{eq:DopProdForm}
\end{align}
where $\langle \mathbf{u},\mathbf{v}\rangle$ is the symplectic form
\be
\langle \mathbf{u},\mathbf{v}\rangle = \tr(u_2 v_1 - u_1 v_2)
\ee
It follows that 
\begin{align}
\left(D_{\mathbf{u}}\right)^k & = D_{k \mathbf{r}}
\\
\intertext{for every integer $k$ and all $\mathbf{u} \in \mathbb{F}_d$.  In particular}
\left(D_{\mathbf{u}}\right)^p & = 1
\end{align}
So the displacement operators are all order $p$ (apart from $D_{\boldsymbol{0}}$).  

It can be seen from the above that the set of operators $\{\omega^{m} D_{\mathbf{u}}\colon \mathbf{u} \in \mathbb{F}_d \text{, } m \in \mathbb{Z}_p\}$ constitutes a group, which is what we are calling the Galoisian variant of the Weyl-Heisenberg group, and which we will denote $\mathcal{W}_d$.  

Many authors define $\mathcal{W}_d$ in a way that may look, on the face of it, rather different, as a tensor product of $n$ copies of the ordinary variant of the Weyl-Heisenberg group in dimension $p$ (the ``many-particle'' definition, in Gross's~~\cite{GrossB,Gross,GrossC,GrossP} terminology).  To see that this definition is in fact equivalent to ours let $e_r$ be any basis for the field $\mathbb{F}_d$, and let $\bar{e}_r$ be the dual basis.  So
\be
\tr(e_r \bar{e}_s ) = \delta_{r,s}
\ee
for all $r$, $s$.  For arbitrary $x\in \mathbb{F}_d$ let $x_r = \tr(x \bar{e}_r)$ (respectively $\bar{x}_r=\tr(x e_r)$) be its expansion coefficients relative to the basis $e_r$ (respectively $\bar{e}_r$).  So
\be
x = \sum_{r=1}^{n} x_r e_r = \sum_{r=1}^{n} \bar{x}_r \bar{e}_r
\ee
Let $\mathcal{H}_p$ (respectively $\mathcal{H}_d$) be $p$-dimensional (respectively $d$-dimensional) Hilbert space, and let $|0\rangle, \dots , |p-1\rangle$ be the standard basis for $\mathcal{H}_p$.  Let $D^p_{\mathbf{u}}$ be the ordinary Weyl-Heisenberg displacement operators in dimension $p$.  So
\be
D^p_{\mathbf{u}} |k\rangle= \tau^{u_1u_2+2 k u_2}|k+u_1\rangle
\ee
 Finally, let $S\colon \mathcal{H}_d \to \mathcal{H}_p \otimes \dots \mathcal{H}_p$ be the linear isomorphism whose action on the basis states is
\be
S |x\rangle = |x_1\rangle \otimes \dots \otimes |x_n\rangle
\ee
Then 
\be
D_{\mathbf{u}}=D_{(u_1,u_2)}=S^{-1} \left( D^p_{(u_{11},\bar{u}_{21})} \otimes \dots \otimes D^{p}_{(u_{1n},\bar{u}_{2n})} \right)S 
\label{eq:WdTermsWpTensored}
\ee
for all $\mathbf{u}$ (where $u_{1r}=\tr(u_1 \bar{e}_r)$,  $\bar{u}_{2r}=\tr(u_2e_r)$).  This defines an isomorphism of the groups  $\mathcal{W}_p \otimes \dots \otimes \mathcal{W}_p$ and $\mathcal{W}_d$.  Note, however, that it is not a \emph{natural} isomorphism (because there are many pairs of dual bases for the field $\mathbb{F}_d$).

These definitions (the one we gave earlier, and the one which consists in identifying $\mathcal{W}_d$ with the tensor product $\mathcal{W}_p \otimes \dots \otimes \mathcal{W}_p$) both have their advantages.  The tensor product definition is, perhaps, the  more natural of the two in a situation such as occurs when one has a register consisting of $n$ different $p$-dits. In a situation like that the tensor product structure  reflects the actual physics.  Also the tensor product definition may be found attractive because it makes no  reference to the mathematics of Galois fields, which is unfamiliar to many physicists.    However, those advantages partly disappear if one is interested in a single qudit which doesn't naturally split into a product of $n$ different $p$-dits.   In that case the imposition of an (arbitrary) tensor product structure introduces a needless complication.  Moreover, even in a situation where the system of interest does in fact consist of $n$  distinct $p$-dimensional subsystems one might want to make use of the properties of Galois fields.    In a case like that, given that one is going to have to introduce them eventually, it would seem to  make sense to introduce Galois fields right at the outset, in the definition itself.  In this paper we are going to make very heavy use of Galois fields.  We consequently prefer our definition.

Let us now turn to the definition of the Galoisian variant of the Clifford group.  Here again we have a choice between two alternative definitions.  However, whereas in the case of $\mathcal{W}_d$ it is essentially a matter of taste which  definition one prefers, here the choice is substantive in that the groups defined are non-isomorphic.  Let us begin with what we will call the unrestricted (Galoisian) Clifford group, $\urc$. This is the set of  unitaries $U$ which reshuffle and re-phase the displacement operators, according to the prescription
\be
U D_{\mathbf{u}} U^{\dagger} = e^{ig(\mathbf{u})} D_{f(\mathbf{u})}
\ee
for all $\mathbf{u}$ and suitable $U$-dependent functions $f, g$.  Since $D^p_{\mathbf{u}} = 1$ for all $\mathbf{u}$ the phase $e^{ig(\mathbf{p})}$ must be a $p^{\mathrm{th}}$ root of unity. So we must in fact have
\be
U D_{\mathbf{u}} U^{\dagger} = \omega^{h(\mathbf{u})} D_{f(\mathbf{u})}
\ee
for some function $h$ taking values in $\mathbb{Z}_p$.  In other words $\urc$ is  the normalizer of $\mathcal{W}_d$ (\emph{i.e.} it is the set of unitaries $U$ such that $U \mathcal{W}_d U^{\dagger} = \mathcal{W}_d$).  

 It follows from  Eq.~(\ref{eq:DopProdForm}) that, for any such $U$, we must have
\be
 \tau^{2 h(\mathbf{u})+2 h(\mathbf{v})+\langle f(\mathbf{u}),f(\mathbf{v})\rangle} D_{f(\mathbf{u}) + f(\mathbf{v})}
=
UD_{\mathbf{u}}D_{\mathbf{v}} U^{\dagger}
=
\tau^{2 h(\mathbf{u}+\mathbf{v})+\langle \mathbf{u},\mathbf{v}\rangle} D_{f(\mathbf{u}+\mathbf{v})}
\label{eq:CliffDefA}
\ee
for all $\mathbf{u},\mathbf{v}$.
So the function $f$ must have the property
\be
f(\mathbf{u}+\mathbf{v}) = f(\mathbf{u}) + f(\mathbf{v})
\ee
 In particular
\be
f(m\mathbf{u}) = f(\mathbf{u} + \dots + \mathbf{u}) = m f(\mathbf{u})
\ee
for all $m \in \mathbb{Z}_p$.  We will refer to a function with this property as $p$-linear. 

 We  define the restricted (Galoisian) Clifford group $\rc$ to be the subgroup of $\urc$ comprising those unitaries which satisfy the additional requirement that $f$ be, not merely $p$-linear, but $d$-linear in the sense
\be
f(x\mathbf{u}) = x f(\mathbf{u}) 
\label{eq:dLinDef}
\ee
for all $x\in \mathbb{F}_d$.  

Some authors (for example Cormick \emph{et al}~\cite{CormickEtAl}) use the term ``Clifford Group'' to mean $\urc$; others (for example Vourdas~\cite{VourdasD}) use it to mean $\rc$.  Gross~\cite{GrossB,Gross,GrossC,GrossP} has proposed to resolve the ambiguity by  referring to $\urc$ as the ``many-particle'' Clifford group, and to $\rc$ as the ``single-particle'' Clifford group.  We have not adopted Gross's terminology ourselves because it seems to us that the mathematical distinction between the groups $\urc$ and $\rc$ has nothing specially to do with the   distinction between composite and non-composite physical systems.  Suppose we have a non-composite physical system, and suppose we adopt the ``single-particle'' definition (in Gross's terminology) of $\mathcal{W}_d$, as represented by Eqs.~(\ref{eq:GalXopDef}) and~(\ref{eq:GalZopDef}).  Then we will still be faced with the fact that the normalizer of $\mathcal{W}_d$ is $\urc$, not $\rc$.  So $\urc$ is no less relevant to single-particle systems than it is to many-particle ones.  The reason is that the distinction between $p$-linearity and $d$-linearity does not depend on the assumption of a tensor product structure.  It therefore seems to us that our terminology, referring to $\urc$ as the unrestricted group and to $\rc$ as the restricted group, is more appropriate.

In this paper we will focus on the  restricted  group, $\rc$.  
This is not because we consider $\rc$ to be of greater physical and/or mathematical importance than  $\urc$.  Nor is it because we are more interested in single-particle systems than many-particle ones.    It is just that the $d$-linear transformations of  $\rc$ have a number of nice properties  which make them worthy of separate investigation. In particular, they have all the properties needed to prove the existence of minimum uncertainty states in every odd prime power dimension (which was the problem that originally motivated this research).    We have therefore chosen to focus on $\rc$ in this paper, and to defer a discussion of $\urc$ to what we hope will be another, subsequent paper.  

If it is assumed that $f$ is $d$-linear Eq.~(\ref{eq:CliffDefA}) implies
\begin{enumerate}
\item[(a)] The function $f$ is of the form 
\be
f(\mathbf{u}) = F \mathbf{u}
\ee
where $F\in\SL(2,\mathbb{F}_d)$---\emph{i.e.} is a $2\times 2$ matrix
\be
F = \bmt \alpha & \beta \\ \gamma & \delta \emt
\ee
with entries in $\mathbb{F}_d$ and such that $\det F = 1$.
\item[(b)] The function $g$ is of the form 
\be
g(\mathbf{u}) = \langle \boldsymbol{\chi},\mathbf{u} \rangle
\ee
for some fixed vector $\boldsymbol{\chi}\in (\mathbb{F}_d)^2$.
\end{enumerate}
We refer to the matrices in $\SL(2,\mathbb{F}_d)$ as symplectic matrices.

Let us sketch the proof of these statements. It is easily seen that $\langle F \mathbf{u},F\mathbf{v}\rangle=\tr\bigl( \det F (u_2 v_1-u_1 v_2)\bigr)$.  In view of Eq.~(\ref{eq:CliffDefA}) this means
\begin{align}
2h(\mathbf{u}+\mathbf{v}) -2 h(\mathbf{u}) -2h(\mathbf{v}) &= \tr\bigl(( \det F -1)(u_2 v_1-u_1 v_2)\bigr)
\label{eq:CliffDefB}
\\
\intertext{(mod $p$).  Interchanging $\mathbf{u}$ and $\mathbf{v}$ we also have}
2h(\mathbf{u}+\mathbf{v}) -2 h(\mathbf{u}) -2h(\mathbf{v}) &= -\tr\bigl(( \det F -1)(u_2 v_1-u_1 v_2)\bigr)
\end{align}
(mod $p$).  It follows that $\det F = 1$.
Setting $\det F=1$ in Eq.~(\ref{eq:CliffDefB}) we deduce that 
 $h$ is $p$-linear.   In terms of the dual bases introduced earlier we have
\be
\mathbf{u} = (u_1,u_2) = \sum_{r=1}^{n}\bigl( u_{1r}  (e_r, 0) + u_{2r} (0,e_r)\bigr)
\ee
where $u_{jr} = \tr(u_j \bar{e}_r)$. So if we define
\begin{align}
\boldsymbol{\chi} = \left(-\sum_{r=1}^{n} h\bigl( (0,e_r)\bigr)\bar{e}_r,\sum_{r=1}^{n} h\bigl( (e_r,0)\bigr) \bar{e}_r \right)
\end{align}
it follows from the $p$-linearity of $h$ that 
\be
h(\mathbf{u}) = \langle \boldsymbol{\chi},\mathbf{u}\rangle
\ee

This establishes that for every $U\in \rc$ there is a matrix $F\in \SL(2,\mathbb{F}_d)$ and vector $\boldsymbol{\chi}\in(\mathbb{F}_d)^2$ such that
\be
U D_{\mathbf{u}} U^{\dagger} = \omega^{\langle \boldsymbol{\chi},\mathbf{u}\rangle } D_{F\mathbf{u}}
\label{eq:cliffAction}
\ee
for all $\mathbf{u}$.
The converse is also true:  for each $F$ and $\boldsymbol{\chi}$ there is a corresponding  $U\in\rc$.  
In fact, it can be shown that for each $F\in\SL(2,\mathbb{F}_d)$ there is a unitary $U_F$ such that 
\be
U^{\vphantom{\dagger}}_F D^{\vphantom{\dagger}}_{\mathbf{u}} U^{\dagger}_F = D^{\vphantom{\dagger}}_{F\mathbf{u}}
\ee 
(in Section~\ref{sec:Faithful} we will give an explicit formula for $U_F$).  So the unitary $U=U_F D_{\boldsymbol{\chi}}$ will transform the displacement operators according to the prescription of Eq.~(\ref{eq:cliffAction}).  Moreover it is easily seen, in view of the irreducibility of the standard representation of $\mathcal{W}_d$ and Schur's lemma, that if $U'$ is any other unitary satisfying Eq.~(\ref{eq:cliffAction}) then $U'=e^{i\theta} U$ for some phase $e^{i\theta}$. In short $\rc$ consists of the set of unitaries of the form
\be
e^{i\theta} U_F D_{\boldsymbol{\chi}}
\ee 
for arbitrary $F\in \SL(2,\mathbb{F}_d)$, $\boldsymbol{\chi} \in (\mathbb{F}_d)^2$, $\theta \in \mathbb{R}$.  
The properties of the displacement operators are easily understood, so in the remainder of this paper most  of the work will go into understanding the properties of the symplectic unitaries $U_F$.

It is also interesting, and for our purposes important, to consider anti-unitary operators $U$ with the property
\be
U D_{\mathbf{u}} U^{\dagger} = e^{ig(\mathbf{u})} D_{f(\mathbf{u})}
\ee
for all $\mathbf{u}$ and suitable functions $f$, $g$.  As before the fact that $D^p_{\mathbf{u}} = 1$ means that $e^{ig(\mathbf{u})}$ must be a $p^{\mathrm{th}}$ root of unity.  So we must have
\be
U D_{\mathbf{u}} U^{\dagger} = \omega^{h(\mathbf{u})} D_{f(\mathbf{u})}
\ee
We will refer to such operators as anti-Clifford anti-unitaries.  The group which consists of  the Clifford unitaries $\in \urc$ together with  the anti-Clifford anti-unitaries just defined we will call the extended (Galoisian) Clifford group $\ure$.  

The description of $\ure$ parallels the discussion in Section $4$ of ref.~\cite{selfB}.  Let $J$ be the matrix
\be
J = \bmt 1 & 0 \\ 0 & -1\emt
\label{eq:JDefNew}
\ee
(so $\det J = -1$) and let $U_J$ be the anti-unitary which acts by complex conjugation in the standard basis:
\be
U_J \left(  \sum_{x\in\mathbb{F}_d} c^{\vphantom{*}}_x |x\rangle \right) = \sum_{x\in\mathbb{F}_d} c^{*}_x |x
\rangle 
\label{eq:UJDefNew}
\ee
Then it is easily seen that
\be
U^{\vphantom{\dagger}}_J D^{\vphantom{\dagger}}_{\mathbf{u}} U^{\dagger}_J = D^{\vphantom{\dagger}}_{J\mathbf{u}}
\ee
for all $\mathbf{u}$.  Now consider the unitary operator $V = UU_J $.   We have
\be
V D_{\mathbf{u}} V^{\dagger} = \omega^{h(J\mathbf{u})} D_{ f(J\mathbf{u})}
\ee
implying that $V\in\urc$. Conversely, if $V\in \urc$ then $ VU_J$ is an anti-Clifford anti-unitary.  So $\ure$ is the disjoint union
\be
\ure= \urc \cup \bigl(\urc U_J \bigr)
\ee

In this paper we will be concerned, not with the full group $\ure$, but with  the subgroup
\be
\re = \rc \cup \bigl( \rc U_J \bigr)
\ee 
We accordingly define $\ESL(2,\mathbb{F}_d)$ to be the group consisting  of all $2\times 2$ matrices
with entries in $\mathbb{F}_d$ and determinant $=\pm 1$. We refer to the matrices in $\ESL(2,\mathbb{F}_d)$ with  determinant $=-1$ as anti-symplectic matrices.  If $\det F=1$ let $U_F$ be the unitary defined earlier, while if $\det F = -1$ let it be the anti-unitary
\be
U_F = U_{FJ} U_J
\ee
Then it is easily seen that $\re$ consists of all operators of the form
\be
U=e^{i\theta} U_F D_{\boldsymbol{\chi}}
\ee
with $F \in \ESL(2,\mathbb{F}_d)$, $\boldsymbol{\chi}\in (\mathbb{F}_d)^2$ and $e^{i \theta}$ an arbitrary phase. If $\det F=1$ then $U$ is a Clifford unitary; if $\det F =-1$ then it is an anti-Clifford anti-unitary.  In either case
\be
U D_{\mathbf{u}} U^{\dagger} = \omega^{\langle \boldsymbol{\chi},\mathbf{u}\rangle} D_{F\mathbf{u}}
\ee 
for all $\mathbf{u}$.

\section{Faithful Representation of $\ESL(2,\mathbb{F}_d)$}
\label{sec:Faithful}
Let $F_1,F_2$ be any two elements of $\ESL(2,\mathbb{F}_d)$.  Then 
\be
U^{\vphantom{\dagger}}_{F_1} U^{\vphantom{\dagger}}_{F_2} D^{\vphantom{\dagger}}_{\mathbf{u}} U^{\dagger}_{F_2}U^{\dagger}_{F_1} = D^{\vphantom{\dagger}}_{F_1F_2\mathbf{u}} = U^{\vphantom{\dagger}}_{F_1F_2} D^{\vphantom{\dagger}}_{\mathbf{u}} U^{\dagger}_{F_1F_2}
\ee
for all $\mathbf{u}$.  It follows from the irreducibility of the standard representation of $\mathcal{W}_d$ and Schur's lemma that
\be
U_{F_1} U_{F_2} = e^{i\theta} U_{F_1F_2}
\ee
for some phase $e^{i\theta}$, meaning  that the map $F\to U_F$  is a projective representation of the group $\ESL(2,\mathbb{F}_d)$ (the metaplectic representation~\cite{Neuhauser,Gross}).  In this section we show that it is possible to choose the phases of the operators $U_F$ in such a way that the representation becomes, not merely projective, but faithful, so that
\be
U_{F_1}U_{F_2} = U_{F_1F_2}
\ee
for all $F_1$, $F_2$.  The fact that this is possible for the unitaries corresponding to the elements of $\SL(2,\mathbb{F}_d)$ has been shown by Neuhauser~\cite{Neuhauser}.  We improve on his result by (a) showing that it extends to the whole of the group $\ESL(2,\mathbb{F}_d)$ and (b) giving explicit formulae which enable one easily to calculate the unitary/anti-unitary corresponding to an arbitrary element of $\ESL(2,\mathbb{F}_d)$.

We begin by establishing some preliminary results.  Let $\mathbb{F}^{*}_d$ (respectively $\mathbb{Z}^{*}_p$) be the set of non-zero elements of $\mathbb{F}_d$ (respectively $\mathbb{Z}_p$).   Let $\theta$ be a primitive element for the field $\mathbb{F}_d$, and for each $x\in \mathbb{F}^{*}_d$ let $\log_{\theta} x$ be the unique integer  in the interval $[0,d-2]$ with the property
\be
x = \theta^{\log_{\theta} x} 
\ee
Let $Q$ (respectively $N$) denote the set of quadratic residues (respectively non-quadratic residues) of $\mathbb{F}_d$:
\begin{align}
Q & = \{x\in \mathbb{F}^{*}_d: x = y^2 \text{ for some } y\in\mathbb{F}_d\}
\\
N & = \{x\in \mathbb{F}^{*}_d: x \neq y^2 \text{ for all } y\in\mathbb{F}_d\}
\end{align}
Clearly $x\in Q$ (respectively $x \in N$) if and only if $\log_{\theta}x$ is even (respectively odd).  Consequently $Q$ and $N$ each contain exactly $(d-1)/2$ elements.
Let $l(x)$ be the quadratic character of $\mathbb{F}_x$
\be
l(x) 
=
\begin{cases}
1 \qquad & x\in Q
\\
-1 \qquad & x\in N
\\
0 \qquad & x=0
\end{cases}
\ee
and let $l_p(z)$ be the quadratic character of $\mathbb{Z}_p$ (\emph{i.e.} the Legendre symbol $\genfrac{(}{)}{1 pt}{1}{z}{p}$).  
Note that $l$, like $l_p$,  is multiplicative:  $l(xy) = l(x) l(y)$ for all $x,y\in\mathbb{F}^{*}_d$.   We then have the following result
\begin{lemma}
\label{lem:lAndlp}
Suppose $z\in \mathbb{Z}^{*}_p$.  Then 
\be
l(z) = 
\begin{cases}
1 \qquad & \text{if $d$ is an even power of $p$}
\\
l_p (z) \qquad & \text{if $d$ is an odd power of $p$}
\end{cases}
\ee
\end{lemma}
\begin{proof}
Define
\be
\epsilon = \theta^{\frac{d-1}{p-1}} = \theta^{1+p+ \dots + p^{n-1}}
\ee
Then $\epsilon^{p} = \epsilon$, implying that $\epsilon \in \mathbb{Z}_p$.  Moreover $\epsilon^{r} = 1$ if and only if $r=0$ (mod $p-1$), implying that $\epsilon$ has multiplicative order $p-1$.  So $\epsilon$ is a primitive element for $\mathbb{Z}_p$.  By construction
\be
\log_\theta z = (1+p+ \dots + p^{n-1}) \log_{\epsilon} z
\ee
for all $z\in \mathbb{Z}^{*}_p$.
If $n$ is even then $(1+p+ \dots + p^{n-1})$ is even, implying that $\log_\theta z$ is even for all $z\in \mathbb{Z}^{*}_p$.  Consequently $l(z)=1$ for all $z\in \mathbb{Z}^{*}_p$.
If, on the other hand, $n$ is odd then $(1+p+ \dots + p^{n-1})$ is odd, implying that $\log_\theta z$ is even if and only if $\log_{\epsilon} z$ is even.  Consequently $l(z) = l_{p}(z)$ for all  $z\in \mathbb{Z}^{*}_p$.
\end{proof}
It is also convenient to define the quantity
\be
\tilde{l}(x)
=
\begin{cases}
-i^{-\frac{n(p+3)}{2}} l(x) \qquad & x\neq 0
\\
1 \qquad & x = 0
\end{cases}
\ee
It has the property
\begin{lemma}
\label{lem:tilLConj}
For all $x\in\mathbb{F}_d$
\be
\bigl(\tilde{l}(x)\bigr)^{*} = \tilde{l}(-x)
\ee
\end{lemma}
\begin{proof}
Suppose $d$ is an even power of $p$.  Then it follows from Lemma~\ref{lem:lAndlp} that   $l(-1) = 1$.  Consequently $l(-x) = l(-1)l(x) = l(x)$.  Moreover, the fact that $n$ and $p+3$ are both even means that $i^{\frac{n(p+3)}{2}}$ is real.  The claim is now immediate.

Suppose, on the other hand, that $d$ is an odd power of p.  Then in view of Lemma~\ref{lem:lAndlp} and a standard result of number theory (see, for example, Hardy and Wright~\cite{Hardy}) we have 
\be
l(-x) = l_p(-1) l(x) = (-1)^{\frac{p+3}{2}} l(x) = (-1)^{-\frac{n(p+3)}{2}} l(x)=i^{-n (p+3)} l(x)
\ee
The claim now follows.
\end{proof}
We will also need the following  Gaussian sum formula
\begin{lemma}
\label{lem:gauss}
For all $a\in\mathbb{F}^{*}_d$ and $b\in\mathbb{F}_{d}$ 
\be
\sum_{x\in \mathbb{F}_d} \tau^{\tr(a x^2 + b x)} =\sqrt{d} \tilde{l}(a) \tau^{-\tr\left(\frac{b^2}{4 a}\right)}
\ee
\end{lemma}
\begin{remark}
The fact that we work in terms of $\tau$ rather than $\omega$ means that we can give a single formula which covers both the cases $p=\pm 1$ (mod $4$).
\end{remark}
\begin{proof}
Using Theorem 5.15 in Lidl and Niederreiter~\cite{LidNie} in conjunction with the fact that $\sum_{x\in\mathbb{F}_d} \omega^{\tr (x)} =0$ we find
\begin{align}
\sum_{x\in\mathbb{F}_d} \omega^{\tr(a x^2)}
& =
\begin{cases}
1+ 2 \sum_{x\in Q} \omega^{\tr(x)} \qquad & \text{if $a\in Q$}
\\
1+ 2 \sum_{x\in N} \omega^{\tr(x)} \qquad & \text{if $a\in N$}
\end{cases}
\\
& = l(a) \sum_{x\in\mathbb{F}_d} l(x) \omega^{\tr(x)}
\\
& =
\begin{cases}
(-1)^{n-1} l(a) \sqrt{d}  \qquad & \text{$p=1$ (mod $4$)}
\\
(-1)^{n-1} i^n l(a) \sqrt{d} \qquad & \text{$p=-1$ (mod $4$)}
\end{cases}
\end{align}
Since $(d+1)/2$ is the multiplicative inverse of $2$ considered as an element of $\mathbb{Z}_p$ we have $\tau^{\tr(ax^2)} = \omega^{\tr(ax^2/2)}$ and, consequently,
\be
\sum_{x\in \mathbb{F}_d} \tau^{\tr(ax^2)} =
\begin{cases}
(-1)^{n-1} l(a/2) \sqrt{d}  \qquad & \text{$p=1$ (mod $4$)}
\\
(-1)^{n-1} i^n l(a/2) \sqrt{d} \qquad & \text{$p=-1$ (mod $4$)}
\end{cases}
\ee

Suppose, now, that $n$ is even.  Then it follows from Lemma~\ref{lem:lAndlp} that $l(a/2)= l(a)$.  Moreover
\be
\tilde{l}(a) = 
\begin{cases}
(-1)^{n-1} l(a) \qquad & \text{$p=1$ (mod $4$)}
\\
(-1)^{n-1} i^n l(a) \qquad & \text{$p=-1$ (mod $4$)}
\end{cases}
\ee
Consequently
\be
\sum_{x\in \mathbb{F}_d} \tau^{\tr(ax^2)} =\sqrt{d}
\tilde{l}(a)
\ee
irrespective of the value of $p$.

Suppose, on the other hand, that $n$ is odd.  Then it follows from Lemma~\ref{lem:lAndlp} and a standard result of number theory (see, for example, Hardy and Wright~\cite{Hardy}) that
\be
l(a/2) = l_p(2) l(a) = 
\begin{cases}
l(a) \qquad & \text{$p=\pm 1$ (mod $8$)}
\\
-l(a) \qquad & \text{$p=\pm 3$ (mod $8$)}
\end{cases}
\ee
So
\be
\sum_{x\in \mathbb{F}_d} \tau^{\tr(ax^2)} =
\begin{cases}
 l(a) \sqrt{d}  \qquad & \text{$p=1$ (mod $8$)}
\\
- i^n l(a) \sqrt{d} \qquad & \text{$p=3$ (mod $8$)}
\\
-l(a) \sqrt{d}   \qquad & \text{$p=5$ (mod $8$)}
\\ 
i^n l(a) \sqrt{d}  \qquad & \text{$p=7$ (mod $8$)}
\end{cases}
\ee
Comparing with the definition of $\tilde{l}(a)$ we again find
\be
\sum_{x\in \mathbb{F}_d} \tau^{\tr(ax^2)} =\sqrt{d}
\tilde{l}(a)
\ee
irrespective of the value of $p$.  

Finally, if $b\neq 0$, we have
\be
\sum_{x\in \mathbb{F}_d} \tau^{\tr(ax^2+bx)} = 
\sum_{x\in \mathbb{F}_d} \tau^{\tr\left(a\left(x+\frac{b}{2a})^2\right)-\frac{b^2}{4 a}\right)}
=
\sqrt{d}\tilde{l}(a) \tau^{-\tr\left(\frac{b^2}{4 a}\right)}
\ee
\end{proof}
We are now ready to establish the main result of this section.  Let 
\be
F = \bmt \alpha & \beta \\ \gamma & \delta \emt
\ee
be any matrix $\in \ESL(2,\mathbb{F}_d)$.  If $\det F = 1$ define
\be
U_F 
=
\begin{cases}
l(\alpha) \sum_{x\in\mathbb{F}_d} \tau^{\tr(\alpha \gamma x^2)} |\alpha x\rangle \langle x | \qquad & \text{if $\beta = 0$}
\\ \vphantom{\Biggl(}
\frac{\tilde{l}(-\beta)}{\sqrt{d}} \sum_{x,y\in \mathbb{F}_d} \tau^{\tr\left( \beta^{-1} (\alpha y^2 - 2 x y + \delta x^2)\right)} |x\rangle \langle y|
\qquad & \text{if $\beta \neq 0$}
\end{cases}
\label{eq:symRepDef1}
\ee
If, on the other hand, $\det F = -1$  define
\be
U_{F} = U_{FJ} U_J
\label{eq:symRepDef2}
\ee
where $U_{FJ}$ is the unitary  defined by Eq.~(\ref{eq:symRepDef1}), $J$ is the anti-symplectic defined by Eq.~(\ref{eq:JDefNew}), and $U_J$ is the anti-unitary defined by Eq.~(\ref{eq:UJDefNew}) (note that $\det (FJ) = 1$, so this definition makes sense).  

We  have
\begin{theorem}
For all $F\in \ESL(2,\mathbb{F}_d)$ the operator $U_F$ is unitary if $\det F= 1$ and anti-unitary if $\det F = -1$.  We have
\begin{enumerate}
\item For all $F\in \ESL(2,\mathbb{F}_d)$ and $\mathbf{u} \in \mathbb{F}^2_d$
\be
U^{\vphantom{\dagger}}_F D^{\vphantom{\dagger}}_{\mathbf{u}} U^{\dagger}_F = D^{\vphantom{\dagger}}_{F\mathbf{u}}
\label{eq:symTrans}
\ee
\item For all $F, F'\in \ESL(2, \mathbb{F}_d)$
\begin{align}
U_{F} U_{F'} & = U_{FF'}
\label{eq:frep1}
\\
U^{\vphantom{\dagger}}_{F^{-1}} &= U^{\dagger}_F
\label{eq:frep2}
\end{align}
\end{enumerate}
\end{theorem}
\begin{proof}
Suppose, first of all, that $\det F = 1$ and $\beta = 0$.  Then 
\begin{align}
U^{\vphantom{\dagger}}_F D_{\mathbf{u}} U^{\dagger}_F 
& = 
\sum_{x,y\in\mathbb{F}_d}\tau^{\tr(\alpha\gamma(x^2 - y^2))} \langle x | D_{\mathbf{u}}| y \rangle |\alpha x\rangle\langle \alpha y|
\\
& =\sum_{x,y\in\mathbb{F}_d}\tau^{\tr(\alpha^{-1}\gamma(x^2 - y^2))} \tau^{\tr(u_1 u_2 + 2 \alpha^{-1} u_2 y)} \delta_{x,y+\alpha u_1} | x \rangle \langle y |
\\
& =\sum_{x,y\in\mathbb{F}_d}\tau^{\tr(\alpha u_1(\gamma u_1 +\delta u_2))} \omega^{\tr((\gamma u_1 + \delta u_2) y)}|y+\alpha u_1\rangle \langle y|
\\
& = D_{F \mathbf{u}}
\end{align}
which establishes Eq.~(\ref{eq:symTrans}) for this case.  Setting $\mathbf{u} = \boldsymbol{0}$ we also see that $U_F$ is unitary.

Suppose, next, that $\det F = 1$ and $\beta \neq 0$.  Then
\begin{align}
U^{\vphantom{\dagger}}_F D_{\mathbf{u}} U^{\dagger}_F 
& = 
\frac{1}{d}\sum_{\substack{x_1,x_2,\\ y_1,y_2\in \mathbb{F}_d}}
\tau^{\tr(\beta^{-1} (\alpha (y_1^2 -y_2^2) - 2 (x_1 y_1 - x_2 y_2 ) + \delta(x_1^2 - x_2^2)))} \langle y_1 | D_{\mathbf{u}} | y_2\rangle |x_1\rangle \langle x_2|
\nonumber
\\
& = \frac{1}{d}\sum_{\substack{x_1,x_2,\\ y_1,y_2\in \mathbb{F}_d}} \tau^{\tr(\beta^{-1} (\alpha (y_1^2 -y_2^2) - 2 (x_1 y_1 - x_2 y_2 ) 
+ \delta(x_1^2 - x_2^2))+u_1 u_2 +2 u_2 y_2)}
\nonumber
\\
& \hspace{2.5 in} \times \delta_{y_1,y_2+u_1} |x_1\rangle \langle x_2|
\nonumber
\\
& = \sum_{x_1,x_2\in\mathbb{F}_d}
\left(\frac{1}{d}\sum_{y_2\in \mathbb{F}_d} \tau^{\tr(-2\beta^{-1}(x_1 -x_2-\alpha u_1-\beta u_2)y_2)}
\right) 
\nonumber
\\
&  \hspace{1 in} \times\tau^{\tr(\beta^{-1}(\delta(x_1^2-x_2^2)-2 u_1 x_1 +\alpha u_1^2 +\beta u_1 u_2))} |x_1\rangle \langle x_2|
\nonumber
\\
&  =
\sum_{x_1,x_2\in \mathbb{F}_d} \tau^{\tr(\beta^{-1}(\delta(x_1^2-x_2^2)-2 u_1 x_1 +\alpha u_1^2 +\beta u_1 u_2))} \delta_{x_1,x_2+\alpha u_1 +\beta u_2} |x_1\rangle \langle x_2|
\nonumber
\\
& = \sum_{x\in \mathbb{F}_d} \tau^{\tr((\alpha u_1+\beta u_2)(\gamma u_1 + \delta u_2))}\omega^{\tr((\gamma u_1 + \delta u_2) x)}|x+\alpha u_1 + \beta u_2\rangle \langle x |
\nonumber
\\
& = D_{F\mathbf{u}}
\end{align}
which establishes Eq.~(\ref{eq:symTrans}) and the fact that $U_F$ is unitary for this case also.

It remains to consider the case when $\det F = -1$.  We have
\be
U^{\vphantom{\dagger}}_J D^{\vphantom{\dagger}}_{\mathbf{u}} U_J^{\dagger} 
= D^{\vphantom{\dagger}}_{J\mathbf{u}}
\ee
In view of the foregoing it follows that
\be
U^{\vphantom{\dagger}}_F D^{\vphantom{\dagger}}_{\mathbf{u}} U^{\dagger}_F 
= 
U^{\vphantom{\dagger}}_{FJ} U^{\vphantom{\dagger}}_J D^{\vphantom{\dagger}}_{\mathbf{u}} U_J^{\dagger} U^{\dagger}_{FJ}
=
D^{\vphantom{\dagger}}_{F\mathbf{u}}
\ee

We now turn to the proof that the representation is faithful.  Let
\be
F_1  = \bmt \alpha_1 & \beta_1 \\ \gamma_1 & \delta_1 \emt
\qquad \qquad
F_2  = \bmt \alpha_2 & \beta_2 \\ \gamma_2 & \delta_2 \emt
\ee
be two matrices $\in \ESL(2,\mathbb{F}_d)$, and let
\be
F  =\bmt \alpha & \beta \\ \gamma & \delta \emt
\ee
be the product $F_1 F_2$.

To start with suppose that $\det F_1 = \det F_2 = 1$. 
There are four cases to consider
\subsection*{Case 1}  $\beta_1= \beta_2 =0$.  Then
\begin{align}
U_{F_1} U_{F_2} 
& = 
l(\alpha_1\alpha_2) \sum_{x_1,x_2\in \mathbb{F}_d} \tau^{\tr(\alpha_1 \gamma_1 x_1^2 + \alpha_2 \gamma_2x^2_2)}\delta_{x_1,\alpha_2 x_2} |\alpha_1 x_1\rangle \langle x_2 |
\nonumber
\\
& = l(\alpha) \sum_{x\in \mathbb{F}_d}\tau^{\tr(\alpha \gamma x^2)} |\alpha x\rangle \langle x |
\nonumber
\\
& = U_F
\end{align}
\subsection*{Case 2}  $\beta_1 = 0$, $\beta_2 \neq 0$.  Then
\begin{align}
U_{F_1} U_{F_2} 
& = 
\frac{l(\alpha_1) \tilde{l}(-\beta_2)}{\sqrt{d}}
\sum_{\substack{x_1, x_2,\\y_2\in \mathbb{F}_d}}
\tau^{\tr(\alpha_1 \gamma_1 x_1^2 +\beta_2^{-1} (\alpha_2 y_2^2 -2 x_2 y_2 + \delta_2 x_2^2))}\delta_{x_1,x_2}|\alpha_1 x_1\rangle \langle y_2 |
\nonumber
\\
& =
\frac{\tilde{l}(-\alpha_1\beta_2)}{\sqrt{d}}
\sum_{x,y\in \mathbb{F}_d} \tau^{\tr\bigl(\alpha^{-1}_1\beta^{-1}_2\bigl((\gamma_1 \beta_2+\delta_1 \delta_2)x^2 -2  x y+\alpha_1 \alpha_2 y^2 \bigr)\bigr)} |x\rangle \langle y|
\nonumber
\\
&= \frac{\tilde{l}(-\beta)}{\sqrt{d}}
\sum_{x,y\in\mathbb{F}_d}\tau^{\tr(\beta^{-1}(\delta x^2 - 2 x y + \alpha y^2)}|x\rangle \langle y |
\nonumber
\\
& = U_F
\end{align}
\subsection*{Case 3} $\beta_1 \neq 0$, $\beta_2 = 0$.  The proof is similar to that of Case 2.
\subsection*{Case 4} $\beta_1 \neq 0$, $\beta_2 \neq 0$.Then
\begin{align}
U_{F_1} U_{F_2} 
& = 
\frac{\tilde{l}(-\beta_1)\tilde{l}(-\beta_2)}{d}
\sum_{\substack{x_1,x_2,\\y_1,y_2 \in \mathbb{F}_d}}
\tau^{\tr\bigl(\beta^{-1}_1(\alpha_1 y_1^2 -2 x_1 y_1 + \delta_1 x_1^2)+\beta^{-1}_2(\alpha_2 y_2^2 -2 x_2 y_2 + \delta_2 x_2^2) \bigr)}
\nonumber
\\
& \hspace{2.5 in} \times \delta_{y_1,x_2} |x_1\rangle \langle y_2|
\nonumber
\\
& = \frac{\tilde{l}(-\beta_1)\tilde{l}(-\beta_2)}{d}
\sum_{x,y\in \mathbb{F}_d} 
\left(\sum_{z\in \mathbb{F}_d} \tau^{\tr\bigl(\beta^{-1}_1\beta^{-1}_2 \beta z^2 - 2 (\beta^{-1}_1 x + \beta^{-1}_2 y)z\bigr)}
\right) 
\nonumber
\\
& \hspace{2 in} \times\tau^{\tr(\beta^{-1}_1\delta_1 x^2 + \beta^{-1}_2 \alpha_2 y^2)} |x\rangle \langle y |
\label{eq:F1F2ProdInterB}
\end{align}
It follows from Lemma~\ref{lem:gauss} that
\begin{align}
&\sum_{z\in \mathbb{F}_d} \tau^{\tr\bigl(\beta^{-1}_1\beta^{-1}_2 \beta z^2 - 2 (\beta^{-1}_1 x + \beta^{-1}_2 y)z\bigr)}
\nonumber
\\
&\hspace{1 in}=
\begin{cases}
d \delta_{x,-\beta^{\vphantom{-1}}_1\beta^{-1}_2 y} \qquad & \beta=0
\\
\sqrt{d} \tilde{l}(\beta^{-1}_1\beta^{-1}_2 \beta) \tau^{-\tr\left(\beta^{-1}\beta_1\beta_2 (\beta^{-1}_1 x+\beta^{-1}_2 y)^2
\right)} \qquad & \beta\neq 0
\end{cases}
\end{align}
 So if $\beta=0$ 
\begin{align}
U_{F_1}U_{F_2} & = 
\tilde{l}(-\beta_1) \tilde{l}(-\beta_2) \sum_{y\in \mathbb{F}_d} \tau^{\tr\bigl(\beta^{\vphantom{-1}}_1\beta^{-1}_2(\beta^{-1}_2\delta^{\vphantom{-1}}_1+\beta^{-1}_1 \alpha^{\vphantom{-1}}_2)y^2\bigr)} |-\beta^{\vphantom{-1}}_1\beta^{-1}_2 y\rangle \langle y|
\label{eq:F1F2ProdInterA}
\end{align}
To evaluate this expression observe that if $\beta=0$
\be
\bmt \alpha_1 & \beta_1 \\ \gamma_1 & \delta_1\emt
=
\bmt \alpha & 0 \\ \gamma & \alpha^{-1} \emt 
\bmt \delta_2 & -\beta_2 \\ -\gamma_2 & \alpha_2 \emt
=
\bmt \alpha\delta_2 & -\alpha \beta_2 \\ \gamma \delta_2-\alpha^{-1} \gamma_2
& -\beta_2 \gamma +\alpha^{-1} \alpha_2
\emt
\ee
from which it is easily seen that
\begin{align}
\alpha & = -\beta^{\vphantom{-1}}_1\beta^{-1}_2
\\
\gamma & = -\beta^{-1}_2 \delta^{\vphantom{-1}}_1 -\beta^{-1}_1 \alpha^{\vphantom{-1}}_2
\end{align}
Also, in view of Lemma~\ref{lem:tilLConj},
\be
\tilde{l}(-\beta_1) \tilde{l}(-\beta_2) 
=\tilde{l}(-\beta_1) \bigl( \tilde{l}(\beta_2)\bigr)^{*}
=l(-\beta_1 \beta_2) = l(-\beta^{\vphantom{-1}}_1 \beta^{-1}_2) = l(\alpha)
\ee
Substituting these expressions into Eq.~(\ref{eq:F1F2ProdInterA}) we deduce
\be
U_{F_1}U_{F_2} = l(\alpha) \sum_{y\in\mathbb{F}_d} \tau^{\tr(\alpha\gamma y^2)}|\alpha y\rangle \langle y| = U_F
\ee
Suppose, on the other hand, that $\beta=0$.  Then Eq.~(\ref{eq:F1F2ProdInterB}) becomes
\begin{align}
U_{F_1}U_{F_2} & = 
\frac{1}{\sqrt{d}}\tilde{l}(-\beta^{\vphantom{-1}}_1)\tilde{l}(-\beta^{\vphantom{-1}}_2) \tilde{l}(\beta^{-1}_1 \beta^{-1}_2\beta)
\nonumber
\\
& \hspace{0.5 in} \times
\sum_{x,y\in\mathbb{F}_d} \tau^{\tr\bigl(\beta^{-1}\bigl(\beta^{-1}_2(\beta\alpha_2 - \beta_1)y^2-2 x y + \beta^{-1}_1 (\beta \delta_1 -\beta_2)x^2\bigr)\bigr)}|x\rangle \langle y|
\label{eq:F1F2ProdInterC}
\end{align}
To evaluate this expression observe that 
\begin{align}
\beta \alpha_2 - \beta_1 &= \beta_2 \alpha
\\
\beta \delta_1 - \beta_2 & = \beta_1 \delta
\end{align}
Also
\begin{align}
\tilde{l}(-\beta^{\vphantom{-1}}_1)\tilde{l}(-\beta^{\vphantom{-1}}_2) \tilde{l}(\beta^{-1}_1 \beta^{-1}_2\beta)
&=
\tilde{l}(-\beta^{\vphantom{-1}}_1)\tilde{l}(-\beta^{\vphantom{-1}}_2)
\bigl(\tilde{l}(-\beta^{-1}_1 \beta^{-1}_2\beta)
\bigr)^{*}
\nonumber
\\
& = - i ^{-\frac{n(p+3)}{2}}l(-\beta_1)l(-\beta_2)l(-\beta^{-1}_1\beta^{-1}_2\beta)
\nonumber
\\
& = - i ^{-\frac{n(p+3)}{2}} l(-\beta)
\nonumber
\\
& = \tilde{l}(-\beta)
\end{align}
Substituting these expressions into Eq.~(\ref{eq:F1F2ProdInterC}) we deduce
\be
U_{F_1}U_{F_2}
=\frac{\tilde{l}(-\beta)}{\sqrt{d}} \sum_{x,y\in\mathbb{F}_d}
\tau^{\tr(\beta^{-1}(\alpha y^2 - 2 x y + \delta x^2))}|x\rangle \langle y| = U_F
\ee

We now turn to the case when either or both of $F_1$, $F_2$ is an anti-symplectic matrix.  
 It is readily verified that  for all $F\in \ESL(2,\mathbb{F}_d)$
\begin{align}
U_F U_J & = U_{FJ}
\\
U_J U_F & =  U_{JF}
\label{eq:faithfulInterA}
\end{align}
Now let $F_1$ be symplectic and $F_2$ anti-symplectic. Then in view of the foregoing
\be
U_{F_1} U_{F_2}  = U_{F_1} U_{F_2J} U_J = U_{F_1 F_2 J} U_J = U_{F_1F_2}
\ee
It then follows from this and Eq.~(\ref{eq:faithfulInterA}) that if $F_1$ is anti-symplectic 
\be
U_{F_1} U_{F_2} = U_{F_1 J} U_J U_{F_2} = U_{F_1 J} U_{J F_2}  = U_{F_1 F_2}
\ee
for all $F_2$, symplectic or anti-symplectic.
Finally, we note that Eq.~(\ref{eq:frep2}) is an immediate consequence of Eq.~(\ref{eq:frep1}) and the fact that the $U_F$ are all unitary/ anti-unitary.  
\end{proof}
 
\section{Expressing the $U_F$ in Terms of the $D_{\mathbf{u}}$}
\label{sec:UFtermsDu}
In the last section we gave explicit expressions for the matrix elements of the operators $U_F$ in the standard basis.  However, it is sometimes useful to express them instead as linear combinations of the $D_{\mathbf{u}}$.  The fact that this is possible is an immediate consequence of the fact that
\be
\Tr\left( D^{\dagger}_{\mathbf{u}} D^{\vphantom{\dagger}}_{\mathbf{v}} 
\right) = d \delta_{\mathbf{u},\mathbf{v}}
\ee
which shows that, relative to the Hilbert-Schmidt inner-product $\langle A, B\rangle = \Tr(A^{\dagger} B)$,  the operators $\frac{1}{\sqrt{d}}D_{\mathbf{u}}$ are an orthonormal basis for the space of linear operators on $d$-dimensional Hilbert space (note that we use the symbol $\Tr$ for the ordinary matrix trace, and $\tr$ for the field-theoretic trace).  The following theorem gives the relevant formulae.

\begin{theorem} 
\label{thm:UtermsD}
Let 
\be
F = \bmt \alpha & \beta \\ \gamma & \delta \emt
\ee
be any matrix $\in \SL(2,\mathbb{F}_d)$ (so $\det F = 1$).  Let $t = \Tr(F)$. Then
\begin{enumerate}
\item[(a)] If $t\neq 2$
\be
U_F =
\begin{cases}
\frac{l(t-2)}{d}\sum_{\mathbf{u} \in \mathbb{F}^2_d} \tau^{\langle \mathbf{u},\tilde{F} \mathbf{u} \rangle } D^{\vphantom{\dagger}}_{\mathbf{u}}\qquad & \beta \neq 0
\\
\frac{l(\alpha)}{d}\sum_{\mathbf{u} \in \mathbb{F}^2_d} \tau^{\langle \mathbf{u},\tilde{F} \mathbf{u} \rangle } D^{\vphantom{\dagger}}_{\mathbf{u}}\qquad & \beta = 0
\end{cases}
\ee
where
\be
\tilde{F}=\frac{1}{2-t} F
\ee
\item[(b)]
If $t=2$
\be
U^{\vphantom{\dagger}}_F =
\begin{cases}
\frac{\tilde{l}(-\beta)}{\sqrt{d}} \sum_{r\in\mathbb{F}_d} \tau^{\tr(\beta r^2)} D^{\vphantom{\dagger}}_{\beta r,(1-\alpha)r}
\qquad & \beta \neq 0
\\
\frac{\tilde{l}(\gamma)}{\sqrt{d}} \sum_{r\in\mathbb{F}_d} \tau^{-\tr(\gamma r^2)} D^{\vphantom{\dagger}}_{0,\gamma r}
\qquad & \text{$\beta =0$ and $\gamma \neq 0$}
\\
D^{\vphantom{\dagger}}_{\boldsymbol{0}} \qquad & \beta = \gamma = 0
\end{cases}
\ee
\end{enumerate}
We also have
\begin{enumerate}
\item[(c)]  If $t\neq 2$
\be
\Tr\left(U_F\right) =
\begin{cases}
l(t-2) \qquad & \beta \neq 0
\\
l(\alpha) \qquad & \beta = 0
\end{cases}
\ee
\item[(d)] If $t=2$
\be
\Tr(U_F) = \begin{cases}
\tilde{l}(-\beta) \sqrt{d}  \qquad & \beta \neq 0
\\
\tilde{l}(\gamma) \sqrt{d} \qquad & \text{$\beta=0$ and $\gamma \neq 0$}
\\
d \qquad & \beta= \gamma = 0
\end{cases}
\ee
\end{enumerate}
\end{theorem}

\begin{proof}
Suppose that $\beta=0$.  Then it follows from the definitions of $D_{\mathbf{u}}$, $U^{\vphantom{\dagger}}_F$ that
\begin{align}
\Tr\left(D^{\dagger}_{\mathbf{u}} U^{\vphantom{\dagger}}_F\right) & = 
\sum_{x,y\in\mathbb{F}_d} \langle x |D^{\dagger}_{\mathbf{u}} | y \rangle \langle y | U_F | x \rangle 
\nonumber
\\
& =
l(\alpha) \sum_{x,y\in\mathbb{F}_d} \tau^{\tr\bigl( \alpha \gamma x^2 - 2 u_2 y + u_1 u_2\bigr)} \delta_{x,y-u_1} \delta_{y,\alpha x}
\nonumber
\\
& =
\begin{cases}
l(\alpha) \tau^{\langle \mathbf{u},\tilde{F} \mathbf{u} \rangle}  \qquad & t \neq 2
\\
\sqrt{d} \tilde{l}(\gamma) \tau^{-\tr(\gamma^{-1} u_2^2)}\delta_{u_1,0}  \qquad & \text{$t=2$ and $\gamma \neq 0$}
\\
d \delta_{\mathbf{u},\boldsymbol{0}} \qquad & \text{$t=2$ and $\gamma = 0$}
\end{cases}
\label{eq:UtermsDInterA}
\end{align}
where we used Lemma~\ref{lem:gauss} in the last line.  Suppose, on the other hand, that $\beta\neq 0$.  Then
\begin{align}
\Tr\left(D^{\dagger}_{\mathbf{u}} U^{\vphantom{\dagger}}_F\right) & = 
\sum_{x,y\in\mathbb{F}_d} \langle x |D^{\dagger}_{\mathbf{u}} | y \rangle \langle y | U^{\vphantom{\dagger}}_F | x \rangle 
\nonumber
\\
&=\frac{1}{\sqrt{d}} \tilde{l}(-\beta)\sum_{x,y\in\mathbb{F}_d} \tau^{\tr\bigl(\beta^{-1}(\alpha x^2 - 2 x y+\delta y^2)  - 2 u_2 y+u_1 u_2\bigr)}\delta_{x,y-u_1}
\nonumber
\\
& = \frac{1}{\sqrt{d}} \tilde{l}(-\beta) \sum_{x\in\mathbb{F}_d} \tau^{\tr\bigl(\beta^{-1}(t-2) x^2 + 2(\beta^{-1}(\delta-1) u_1 - u_2) x + \beta^{-1}\delta u_1^2 - u_1 u_2\bigr)}
\nonumber
\\
&=
\begin{cases}
 l(t-2) \tau^{\langle \mathbf{u},\tilde{F} \mathbf{u}\rangle} \qquad & t\neq 2
\\
\sqrt{d} \tilde{l}(-\beta) \tau^{\tr(\beta^{-1} u_1^2)} \delta_{u_2,\beta^{-1}(1-\alpha) u_1} \qquad & t = 2
\end{cases}
\label{eq:UtermsDInterB}
\end{align}
where we again used Lemma~\ref{lem:gauss}.  Statements (c) and (d) of the theorem are special cases of these formulae (since $D_{\boldsymbol{0}}=1$).  Statements (a) and (b) are immediate consequences of them and the fact
\be
U^{\vphantom{\dagger}}_F  = \frac{1}{d} \sum_{\mathbf{u} \in \mathbb{F}^2_d} \Tr\left(D^{\dagger}_{\mathbf{u}} U^{\vphantom{\dagger}}_F\right) D^{\vphantom{\dagger}}_{\mathbf{u}}
\ee
\end{proof}

\section{Eigenvalues and Orders of the (anti-)Symplectic Matrices}
\label{sec:orders}
We now turn to the problem of calculating the eigenvalues and order of an arbitrary symplectic/anti-symplectic matrix
\be
F = \bmt \alpha & \beta \\ \gamma & \delta \emt
\ee
Once these are known it is straightforward to find the order, the eigenvalues and eigenspaces and the roots of the corresponding unitary/anti-unitary $U_F$.

 Let $t=\alpha+\delta$ be the trace (in the ordinary matrix sense) of $F$, and let $\Delta$ be its determinant.  We say that $F$ is 
\begin{enumerate}
\item type $1$ if $t^2 - 4 \Delta \in Q$
\item type $2$ if $t^2 - 4 \Delta \in N$
\item type $3$ if $t^2 - 4 \Delta = 0$.
\end{enumerate}
Suppose, to begin with, that $F$ is type $1$.  Then $F$ is conjugate to  the matrix~\cite{Flammia,selfC,Hump,Gehles}
\be
\bar{F}= \bmt 0 & -\Delta \\ 1 & t \emt
\ee
This matrix is diagonalizable.  In fact
\be
\det (F-\lambda I) = \lambda^2 - t \lambda + \Delta
\ee
So the equation $\det(F-\lambda I)=0$ has the two solutions
\be
\lambda_{\pm} = \frac{t\pm \sqrt{t^2 - 4 \Delta}}{2}
\ee
where $\sqrt{t^2 - 4 \Delta}$ is one of the two elements of $\mathbb{F}_d$ with the property $\left(\sqrt{t^2 - 4 \Delta}\right)^2 = t^2 - 4 \Delta$ (guaranteed to exist because of the assumption that $t^2 - 4\Delta \in Q$) and where, as usual, division by $2$ means multiplication by the multiplicative inverse of $2$ considered as an element of $\mathbb{F}_d$.  It is easily verified that
\begin{align}
\Delta &= \lambda_{+}\lambda_{-}
\\
t & = \lambda_{+} + \lambda_{-}
\end{align}
This means we can write
\begin{align}
\lambda_{+} & = \theta^r
\\
\lambda_{-} & = \Delta \theta^{-r}
\end{align}
where $r = \log_{\theta} \lambda_{+}$.  The fact that $(\Delta \theta^{-r} -\theta^r)^2 = t^2 - 4 \Delta \in Q$ means that $\Delta \theta^{-r} \neq \theta^r$.  So the matrix
\be
S= \bmt \frac{\theta^{-r}}{\Delta \theta^{-r} -\theta^r} & \frac{1}{\Delta \theta^{-r} -\theta^r} \\ \Delta \theta^r & 1
\emt
\ee
is a well-defined element of $\ESL(2,\mathbb{F}_d)$.  It is straightforward to confirm that 
\be
S \bar{F} S^{-1} = \bmt \theta^r & 0 \\ 0 & \Delta \theta^{-r} \emt
\ee
So the problem of finding the order of $F$ reduces to the problem of finding the smallest positive integer $m$ such that
\be
\bmt \theta^r & 0 \\ 0 & \Delta \theta^{-r} \emt^m = \bmt 1 & 0 \\ 0 & 1\emt
\ee
If $\Delta = 1$ we immediately deduce that
\be
\ord (F) = \frac{d-1}{[r,(d-1)]}
\ee 
where $[r,(d-1)]$ is the greatest common divisor of $r$ and $d-1$.  Suppose, on the other hand, that $\Delta=-1$.  Then the order of $F$ is the smallest even positive integer $m$ such that $\theta^m = 1$.  So
\be
\ord (F) =\begin{cases}  \frac{(d-1)}{[r,(d-1)]}  \qquad & \text{if $\frac{(d-1)}{[r,(d-1)]}$ is even}
\\\vphantom{\biggl(}
\frac{2(d-1)}{[r,(d-1)]}  \qquad & \text{if $\frac{(d-1)}{[r,(d-1)]}$ is odd}
\end{cases}
\ee
We can write this formula more compactly if we observe
\be
\left[r,\frac{d-1}{2}\right] = \begin{cases} [r,d-1] \qquad & \text{if $\frac{(d-1)}{[r,(d-1)]}$ is even}
\\ \vphantom{\biggl(}
\frac{1}{2} [r,d-1] \qquad & \text{if $\frac{(d-1)}{[r,(d-1)]}$ is odd}
\end{cases}
\ee
from which it follows
\be
\ord (F) = \frac{d-1}{[r,\frac{d-1}{2}]}
\ee

Let us note that for all $r$ in the interval $0\le r <d-1$ the matrix
\be
\bar{F} = \bmt 0 & -\Delta \\ 1 & \theta^{r} + \Delta \theta^{-r} \emt
\ee
is a well-defined element of $\ESL(2,\mathbb{F}_d)$.  However, for some values of $r$ the matrix is type $3$ and not diagonalizable.  It is easily seen that the  values of $r$ for which $\bar{F}$ is a type $1$ matrix with eigenvalues $\theta^r$, $\Delta \theta^{-r}$ are
\begin{align}
&\left\{ r\colon 1 \le r < d-1, \  r \neq \frac{d-1}{2} \right\} & & \text{if $\Delta = 1$}
\\
&\left\{ r\colon 1 \le r < d-1, \  r \neq \frac{d-1}{4}, \frac{3(d-1)}{4} \right\}&& \text{if $\Delta = -1$ and $d=1$ (mod $4$)}
\\
&\left\{ r\colon 1 \le r < d-1\right\}&& \text{if $\Delta = -1$ and $d=3$ (mod $4$)}
\end{align}
In particular the maximum order of a type $1$ matrix is $d-1$, the maximum being achieved by (for example) matrices for which $t = \theta+\Delta \theta^{-1}$.  

Suppose, next, that $F$ is type $2$.  Then $F$ is again conjugate to  the matrix~\cite{Flammia,selfC,Hump,Gehles}
\be
\bar{F}= \bmt 0 & -\Delta \\ 1 & t \emt
\ee
However the fact that the discriminant $t^2 - 4 \Delta \in N$ means that the equation
\be
\det(F-\lambda I) = \lambda^2 -t \lambda + \Delta = 0
\label{eq:charEq}
\ee
has no solutions in the field $\mathbb{F}_d$. We can deal with this problem in the same way that we would if we had a matrix defined over the reals whose characteristic equation had a negative discriminant:  namely, we can  go to a suitable extension field, in this case $\mathbb{F}_{d^2}$.  
Let $\bar{\theta}$ be  a primitive element for $\mathbb{F}_{d^2}$.  We make the identification
\be
\mathbb{F}_d = \{x\in \mathbb{F}_{d^2} \colon x^{d-1} = 1\}
\ee
$\mathbb{F}^{*}_d$ thus consists of all powers of $\bar{\theta}^{d+1}$.  It follows that $\bar{\theta}^{d+1}$ is a primitive element for $\mathbb{F}_d$ which, without loss of generality, we can identify with $\theta$.

To define the square root of an arbitrary element $x \in N$ observe that $x = \bar{\theta}^{k (d+1)}$ for some odd integer $k$.  So if we define $\sqrt{x} = \bar{\theta}^{\frac{k(d+1)}{2}}$ we will have $x=\left(\sqrt{x}\right)^2$.  Observe that, since $\bar{\theta}^{\frac{d^2-1}{2}} = -1$ and $k$ is odd,
\be
\left(\sqrt{x}\right)^{d} = \bar{\theta}^{\frac{k(d^2-1)}{2}} \bar{\theta}^{\frac{k(d+1)}{2}} = - \sqrt{x}
\label{eq:rtdPow}
\ee

Now let
\be
\lambda_{\pm} = \frac{t\pm \sqrt{t^2 - 4 \Delta}}{2}
\ee
be the solutions of Eq.~(\ref{eq:charEq}).  Using the fact that $(x\pm y)^d = x^d \pm y^d$ for all $x,y\in\mathbb{F}_{d^2}$ we deduce
\be
\lambda^{d}_{\pm} = \frac{t^d\pm \left(\sqrt{t^2 - 4 \Delta}\right)^d}{2^d}
\ee
In view of Eq.~(\ref{eq:rtdPow}) and the fact that $x^d = x$ for all $x\in \mathbb{F}_d$ it follows
\be
\lambda^{d}_{\pm} = \lambda_{\mp}
\ee
So 
\begin{align}
\lambda_{+} & = \bar{\theta}^{k}
\\
\lambda_{-} & = \bar{\theta}^{kd}
\end{align}
for some integer $k$.  Since $\lambda_{+}\lambda_{-} = \Delta$, and taking account of the fact that $\bar{\theta}^{-\frac{d^2-1}{2}}=-1$, we must have
\be
k(d+1) = \begin{cases} 0 \qquad & \text{if $\Delta=1$}
\\
\frac{d^2-1}{2}  \qquad & \text{if $\Delta=-1$}
\end{cases}
\ee
mod $(d^2-1)$.  So if we define $\eta = \bar{\theta}^{\frac{d-1}{2}}$
\begin{align}
\lambda_{+} & = \eta^r
\\
\lambda_{-} & = \eta^{r d} = (-1)^r \eta^{-r}
\end{align}
for some integer r   which is in the range $1\le r < 2(d+1)$, and which is even if $\Delta =1$ and odd if $\Delta=-1$.  

Since $t^2-4\Delta\neq 0$ the eigenvalues are distinct and so the matrix
\be
S = \bmt \frac{1}{(-1)^r-\eta^{2r}} & \frac{\eta^r}{(-1)^r-\eta^{2r}} \\ (-1)^r \eta^{r} & 1\emt
\ee
is a well-defined element of $\ESL(2, \mathbb{F}_{d^2})$.  It is readily verified that
\be
S \bar{F} S^{-1} = \bmt \eta^r & 0 \\ 0 & (-1)^r \eta^{-r} \emt
\ee
from which we see that the order of $F$ is the smallest positive integer $m$ such that $\eta^{mr}=1$.  Since $\eta^k = 1$ if and only if $k = 0$ (mod $2(d+1)$) we conclude
\be
\ord (F) = \frac{2(d+1)}{[r,2(d+1)]}
\ee

It follows from this that every symplectic matrix has order $\le d+1$, and every anti-symplectic matrix has order $\le 2d+2$.
 Let us  note that there exist both symplectic and anti-symplectic matrices which achieve these upper bounds.  This is a consequence of the more general fact that for every value of $r$ in the interval $1\le r <2(d+1)$ there is a matrix $F\in \ESL(2,\mathbb{F}_d)$ with eigenvalues $\eta^r, (-1)^r \eta^{-r}$.   To prove this observe
\be
\left( \eta^r + (-1)^r \eta^{-r}\right)^d = \eta^{dr } + (-1)^r \eta^{-d r} = \eta^{r}+(-1)^r\eta^{-r } 
\ee
(since $\eta^{dr} = (-1)^r \eta^{-r}$), implying that $\eta^r + (-1)^r \eta^{-r} \in \mathbb{F}_{d}$.  So the matrix
\be
\bar{F}= \bmt 0 & (-1)^{r+1} \\ 1 & \eta^{r} + (-1)^r \eta^{-r} \emt
\ee
 is in $\ESL(2,\mathbb{F}_d)$ and has eigenvalues $\eta^r, (-1)^r \eta^{-r}$.  In particular
\be
\bar{F}=\bmt 0 & 1 \\ -1 & \eta^2 + \eta^{-2} \emt
\ee 
is a symplectic matrix having the maximum order $d+1$ and
\be
\bar{F}=\bmt 0 & 1 \\ 1 & \eta -\eta^{-1} \emt
\ee
is an anti-symplectic matrix having the maximum order $2(d+1)$.  

It remains to consider type $3$ matrices.   A type $3$ matrix is not diagonalizable except in the trivial case when $F = \pm I$.  We have the further  complication  that the conjugacy class of a type $3$ matrix
\be
F = \bmt \alpha & \beta \\ \gamma & \delta \emt
\ee
depends, not only on the values of $t$ and $\Delta$, but also on the values of $\beta$ and $\gamma$.  Finally, we need to consider the cases $d=1$ (mod $4$) and $d=3$ (mod $4$) separately.   Using the results in ref.~\cite{selfC} we give in Table~\ref{tble:type3OrddEq1}  a representative for each of the $12$ conjugacy classes when $d=1$ (mod $4$), while in Table~\ref{tble:type3OrddEq2} we give a representative for each of the $4$ conjugacy classes when $d=3$ (mod $4$).  The last column in each table gives the order of the matrices in the corresponding conjugacy class, calculated using the formula
\be
\bmt x & y \\ 0 & x\emt^r = \bmt x^r & r x^{r-1} y \\ 0 & x^r \emt
\label{eq:type3ordA}
\ee
  
We have thus proved
\begin{theorem}
Let $F$ be any element of $\ESL(2,\mathbb{F}_d)$.  Let $t$ be the trace of $F$ (in the ordinary matrix sense) and let $\Delta$ be the determinant.  Then
\begin{enumerate}
\item If $F$ is type $1$ let 
\be
r = \log_{\theta} \left(\frac{-t + \sqrt{t^2 - 4 \Delta}}{2} 
\right)
\ee
Then the order of $F$ is given by
\be
\ord (F)
=
\begin{cases}
\frac{d-1}{[r,d-1]} \qquad & \text{if $\Delta=1$}
\\ \vphantom{\biggl(}
\frac{d-1}{[r,\frac{d-1}{2}]} \qquad & \text{if $\Delta=-1$}
\end{cases}
\ee
where $[u,v]$ denotes the greatest common divisor of $u$ and $v$.
The allowed values of $r$ are
\begin{align}
&\left\{ r\colon 1 \le r < d-1, \  r \neq \frac{d-1}{2} \right\} & & \text{if $\Delta = 1$}
\\
&\left\{ r\colon 1 \le r < d-1, \  r \neq \frac{d-1}{4}, \frac{3(d-1)}{4} \right\}&& \text{if $\Delta = -1$ and $d=1$ (mod $4$)}
\\
&\left\{ r\colon 1 \le r < d-1\right\}&& \text{if $\Delta = -1$ and $d=3$ (mod $4$)}
\end{align}
In particular the maximum order of a type $1$ matrix is $d-1$, the maximum being achieved by (for example) symplectic matrices for which $t = \theta+ \theta^{-1}$ and anti-symplectic matrices for which $t=\theta-\theta^{-1}$.  
\item If $F$ is type $2$ let $\bar{\theta}$ be a primitive element of $\mathbb{F}_{d^2}$ and let $\eta=\bar{\theta}^{\frac{d-1}{2}}$.  Define
\be
r = \log_{\eta} \left(\frac{-t + \sqrt{t^2 - 4 \Delta}}{2} 
\right)
\ee
Then the order of $F$ is given by
\be
\ord(F) = \frac{2(d+1)}{[r,2(d+1)]}
\ee
Every value of $r$ in the interval $1\le r <2(d+1)$ is possible.  $r$ is even for symplectic  and odd for anti-symplectic matrices.  So the maximum order of a type $2$ symplectic matrix is $d+1$, the maximum being achieved by (for example) matrices for which $t=\eta^2+\eta^{2d}$, while the maximum order for a type $2$ anti-symplectic matrix is $2(d+1)$, the maximum being achieved by  (for example) matrices for which $t=\eta+\eta^d$.  
\item If $F$ is type $3$ the order is as tabulated in Table~\ref{tble:type3OrddEq1} (if $d=1$ (mod $4$)) or Table~\ref{tble:type3OrddEq2} (if $d=3$ (mod $4$)).
\end{enumerate}
\end{theorem}
\begin{table}[htb]
\begin{tabular}{ | p{0.5 in} |  p{0.5 in} | p{1 in} | p{1 in} | p{0.5 in}|}
\hline
\bc $\Delta$ \ec & \bc $t$ \ec & \bc $\beta$, $\gamma$ \ec & \bc representative \ec &  \bc  order \ec \\ \hline
\tbins{0.5}{$1$} &  \tbins{0.5}{$2$} & \tbins{1}{$\beta=\gamma=0$}  & \tbins{1}{$\bmt \mspace{6 mu} 1\mspace{6 mu}  & \mspace{6 mu} 0 \mspace{6 mu} \\ 0 & 1 \emt $} &  \tbins{0.5}{1} 
\\
\tbins{0.5}{$1$} &\tbins{0.5}{$2$}& \tbins{1}{$\beta$ or $\gamma\in Q$} &\tbins{1}{$\bmt \mspace{6 mu} 1 & \mspace{6 mu}1\mspace{6 mu} \\ 0 & \mspace{6 mu}1\mspace{6 mu}\emt $} &  \tbins{0.5}{p}
 \\
\tbins{0.5}{$1$} &\tbins{0.5}{$2$}& \tbins{1}{$\beta$ or $\gamma\in N$}  &\tbins{1}{$\bmt \mspace{6 mu} 1 & \mspace{6 mu}\nu\mspace{6 mu} \\ 0 & \mspace{6 mu}1\mspace{6 mu}\emt $} &  \tbins{0.5}{p} 
\\
\tbins{0.5}{$1$} &  \tbins{0.5}{$-2$} & \tbins{1}{$\beta=\gamma=0$}   &\tbins{1}{$\bmt -1 & 0 \\ 0 & -1 \emt $} &  \tbins{0.5}{2} 
\\
\tbins{0.5}{$1$} &\tbins{0.5}{$-2$}& \tbins{1}{$\beta$ or $\gamma\in Q$} &\tbins{1}{$\bmt -1 & 1 \\ 0 & -1 \emt $} &  \tbins{0.5}{2 p} 
\\
\tbins{0.5}{$1$} &\tbins{0.5}{$-2$}& \tbins{1}{$\beta$ or $\gamma\in N$}  &\tbins{1}{$\bmt -1 & \nu \\ 0 & -1 \emt $} &  \tbins{0.5}{2 p} 
\\
\tbins{0.5}{$-1$} & \tbins{0.5}{$2i$}&  \tbins{1}{$\beta=\gamma=0$} &\tbins{1}{$\bmt \mspace{6 mu} i\mspace{6 mu}  & \mspace{6 mu} 0 \mspace{6 mu} \\ 0 & i \emt $} &  \tbins{0.5}{4}
 \\
\tbins{0.5}{$-1$} & \tbins{0.5}{$2i$}&  \tbins{1}{$\beta$ or $\gamma\in Q$}  &\tbins{1}{$\bmt \mspace{6 mu} i & \mspace{6 mu}1\mspace{6 mu} \\ 0 & \mspace{6 mu}i\mspace{6 mu}\emt $} &  \tbins{0.5}{4 p} 
\\
\tbins{0.5}{$-1$} & \tbins{0.5}{$2i$}& \tbins{1}{$\beta$ or $\gamma\in N$} &\tbins{1}{$\bmt \mspace{6 mu} i & \mspace{6 mu}\nu\mspace{6 mu} \\ 0 & \mspace{6 mu}i\mspace{6 mu}\emt $} &  \tbins{0.5}{4 p} 
\\
\tbins{0.5}{$-1$} & \tbins{0.5}{$-2i$}&  \tbins{1}{$\beta=\gamma=0$}  &\tbins{1}{$\bmt -i & 0 \\ 0 & -i \emt $} &  \tbins{0.5}{4}
\\
\tbins{0.5}{$-1$} & \tbins{0.5}{$-2i$}&  \tbins{1}{$\beta$ or $\gamma\in Q$}  &\tbins{1}{$\bmt -i & 1 \\ 0 & -i \emt $} &  \tbins{0.5}{4 p} 
\\
\tbins{0.5}{$-1$} & \tbins{0.5}{$-2i$}& \tbins{1}{$\beta$ or $\gamma\in N$} &\tbins{1}{$\bmt -i & \nu \\ 0 & -i \emt $} &  \tbins{0.5}{4 p}
 \\
\hline
\end{tabular}
\vspace{0.15 in}
\caption{Orders of Type $3$ matrices when $d= 1 \text{ (mod $4$)}$.  Here $\nu$ is any fixed element of $N$, and $i$ is one of the two solutions to the equation $x^2=-1$ (so $i$ does \emph{not} have its usual meaning---for instance if $d=5$ then $i = 2$ or $4$).  Note, also, that it is not possible for $\beta\in Q$ and $\gamma \in N$, or for $\beta\in N$ and $\gamma\in Q$.  For more details see ref.~\cite{selfC}. }
\label{tble:type3OrddEq1}
\end{table}

\begin{table}[htb]
\begin{tabular}{ | p{0.5 in} |  p{0.5 in} | p{1 in} | p{1 in} | p{0.5 in}|}
\hline
\bc $\Delta$ \ec & \bc $t$ \ec & \bc $\beta$, $\gamma$ \ec & \bc representative \ec &  \bc  order \ec \\ \hline
\tbins{0.5}{$1$} & \tbins{0.5}{$2$}& \tbins{1}{$\beta=\gamma = 0$}&\tbins{1}{$\bmt \mspace{6 mu} 1\mspace{6 mu}  & \mspace{6 mu} 0 \mspace{6 mu} \\ 0 & 1 \emt $} &  \tbins{0.5}{$1$}
 \\
\tbins{0.5}{$1$} & \tbins{0.5}{$2$}& \tbins{1}{$\beta$ or $\gamma \neq 0$} &\tbins{1}{$\bmt \mspace{6 mu} 1\mspace{6 mu}  & \mspace{6 mu} 1 \mspace{6 mu} \\ 0 & 1 \emt $} &  \tbins{0.5}{$p$} 
\\
\tbins{0.5}{$1$} & \tbins{0.5}{$-2$}& \tbins{1}{$\beta=\gamma = 0$} &\tbins{1}{$\bmt -1 & 0 \\ 0 & -1 \emt $} &  \tbins{0.5}{$2$} 
\\
\tbins{0.5}{$1$} & \tbins{0.5}{$-2$}& \tbins{1}{$\beta$ or $\gamma \neq 0$} &\tbins{1}{$\bmt -1  & 1\\ 0 & -1 \emt $} &  \tbins{0.5}{$2 p$}
\\
\hline
\end{tabular}
\vspace{0.15 in}
\caption{Orders of Type $3$ matrices when $d= 3 \text{ (mod $4$)}$.  Note that for these values of $d$ there are no type $3$ anti-symplectics.  For more details see ref.~\cite{selfC}. }
\label{tble:type3OrddEq2}
\end{table}

With these results in hand it is straightforward to find the order, eigenvalues and eigenspaces and roots of the unitary/anti-unitary $U_F$ corresponding to an arbitrary matrix $F\in\ESL(2,\mathbb{F}_d)$.  In fact, let $m=\ord(F)$.  Then $m$ is also the order of $U_F$.  So the eigenvalues of $U_F$ are $e^{\frac{2 i r\pi}{m}}$, with $r=0,1,\dots, m-1$.  The projector onto the eigenspace of $U_F$ with eigenvalue $e^{\frac{2 i r\pi}{m}}$ is
\be
P_r = \frac{1}{m} \sum_{s=0}^{m-1} e^{-\frac{2 i r s\pi}{m}}U_F^s
\ee
The dimension of the eigenspace (possibly zero) is $\Tr(P_r)$.  Finally, suppose we want to find the $s^{\mathrm{th}}$ roots of $U_F$:  \emph{i.e.}\ the set of  (anti)-symplectic unitaries $U_G$ (possibly empty) such that $U^s_G = U_F$.  Clearly $U_G$ has this property if and only if $G^s = F$.  Suppose, for the sake of definiteness, that $F$ is  type~$1$, so that
\be
F = S \bmt \theta^r & 0 \\ 0 & (\det F) \theta^{-r} \emt S^{-1}
\ee
for some $S$, $r$.  Then $G$ is an $s^{\mathrm{th}}$ root if and only if
\be
G = S \bmt \theta^{t} & 0 \\ 0 & \Delta \theta^{-t} \emt S^{-1}
\ee
for some $t$, $\Delta$ such that $st = r$ ( mod $d-1$) and $\Delta^s = \det F$.  The problem of finding the roots of a type $2$ or $3$ matrix is equally straightforward.
\section{SIC-POVMs}
\label{sec:SICs}
The results in the last section are potentially relevant to the problem of constructing SIC-POVMs (symmetric informationally complete positive operator valued measures) in prime dimension.  Every known~\cite{ScottGrassl} SIC-POVM which is covariant under the Weyl-Heisenberg group can be constructed from a fiducial vector which is an eigenvector of a symplectic unitary $U_F$ for which $\Tr(F) = -1$ (mod $d$).  It is therefore of some interest to find a natural basis for the appropriate eigenspace.  The results in the last section enable us to do that in every odd prime dimension (since if $d$ is prime the Galoisian variant of the Weyl-Heisenberg Group considered in this paper coincides with the ordinary variant).  

Throughout this section it will be assumed without comment that $d$ is prime number $\ge 5$ (our analysis does not apply to the case $d=3$) .  If $t=\Tr(F) = -1$ and $F$ is symplectic then  $t^2 - 4 \Delta=-3$.  It is a standard result that $-3\in Q$ (respectively $-3\in N$) if $d=1$ (mod $6$) (respectively $d=5$ (mod $6$)) (see, for example, Hardy and Wright~\cite{Hardy}).  So $F$ is type $1$ if $d=1$ (mod $6$) and type $2$ if $d=5$ (mod $6$).  

Consider first the case $d=6m+1 = 1$ (mod $6$).  Without loss of generality we can take $F$ to be the matrix
\be
F = \bmt \theta^{2m} & 0 \\ 0 & \theta^{-2m } \emt
\ee
Define
\be
G = \bmt \theta & 0 \\ 0 & \theta^{-1} \emt
\ee 
Then $U_F = U^{2m}_G$.  $U_G$ is the permutation matrix
\be
U_G =- \sum_{r=0}^{d-1} |\theta r\rangle \langle r |
\ee
Define
\begin{align}
|\psi_r\rangle & = \frac{1}{\sqrt{d-1}} \sum_{s=1}^{d-1} \sigma^{-r \log_{\theta} s} |s\rangle
\\
\intertext{where $\sigma = e^{\frac{\pi i}{3m}}$ and  $r=0,1,\dots, 6m-1$.  Also define}
|\psi'_0\rangle& = |0\rangle
\end{align}
Then it is readily confirmed 
\begin{align}
U_G | \psi_r \rangle &= \sigma^{r} |\psi_r \rangle
\\
\intertext{for $r=0,1,\dots, d-2$, and}
U_G |\psi'_0\rangle & = |\psi'_0\rangle
\end{align}
Consequently
\begin{align}
U_F | \psi_r \rangle &= \lambda^r |\psi_r \rangle
\\
\intertext{for $r=0,1,\dots, 6m-1$ (where $\lambda = e^{\frac{2\pi i}{3}}$), and}
U_F |\psi'_0\rangle & = |\psi'_0\rangle
\end{align}
It appears from Scott and Grassl's exhaustive numerical investigation~\cite{ScottGrassl} that the fiducial lies in the eigenspace of highest dimension, which is the one with eigenvalue $1$ spanned by the $2m+1$ vectors $|\psi'_0\rangle, |\psi_0\rangle,|\psi_3\rangle, \dots , |\psi_{6m-3}\rangle$.

Now consider the case $d=6m-1 = 5$ (mod $6$).  Without loss of generality we can take $F$ to be the matrix
\be
F = \bmt 0 & -1 \\ 1 & -1\emt
\ee
(the Zauner matrix~\cite{Zauner}).  This is a type $2$ matrix, and therefore not diagonalizable.  However, we can still use the techniques described in the last section to find an order $6m$ symplectic matrix $G$ such that $F = G^{2m}$.  It will turn out that the eigenvalues of the unitary $U_G$ are all non-degenerate, which means that it  provides us with a natural basis for the subspace in which the fiducial is hypothesized to lie.  

To see this observe
\be
F = \bmt 1 & -\frac{\eta^{-4m}}{\eta^{4m}-1} \\ - \eta^{4m} & \frac{\eta^{4m}}{\eta^{4m}-1} \emt \bmt \eta^{4m} & 0 \\ 0 & \eta^{-4m} \emt \bmt 1 & -\frac{\eta^{-4m}}{\eta^{4m}-1} \\ - \eta^{4m} & \frac{\eta^{4m}}{\eta^{4m}-1} \emt^{-1}
\ee
So if we define
\begin{align}
G &= \bmt 1 & -\frac{\eta^{-4m}}{\eta^{4m}-1} \\ - \eta^{4m} & \frac{\eta^{4m}}{\eta^{4m}-1} \emt\bmt \eta^{2} & 0 \\ 0 & \eta^{-2} \emt\bmt 1 & -\frac{\eta^{-4m}}{\eta^{4m}-1} \\ - \eta^{4m} & \frac{\eta^{4m}}{\eta^{4m}-1} \emt^{-1}
\nonumber
\\
& = \bmt \gamma_{m+1} & -\gamma_{1} \\ \gamma_1 & \gamma_{m-1} \emt
\label{eq:GdEq5mod6}
\end{align}
where
\be
\gamma_r = \frac{\eta^{2r} - \eta^{-2r}}{\eta^{2m} - \eta^{-2m}}
\ee
the matrix $G$ will have the desired property $F= G^{2m}$.  Note that we can use the same trick to calculate an arbitrary power of $G$:
\begin{align}
G^r &= \bmt 1 & -\frac{\eta^{-4m}}{\eta^{4m}-1} \\ - \eta^{4m} & \frac{\eta^{4m}}{\eta^{4m}-1} \emt\bmt \eta^{2r} & 0 \\ 0 & \eta^{-2r} \emt\bmt 1 & -\frac{\eta^{-4m}}{\eta^{4m}-1} \\ - \eta^{4m} & \frac{\eta^{4m}}{\eta^{4m}-1} \emt^{-1}
\nonumber
\\
& = \bmt \gamma_{m+r} & -\gamma_r \\ \gamma_{r} & \gamma_{m-r} \emt
\end{align}
Note also that, for all $r$,
\begin{align}
\gamma^d_r &= \frac{(\eta^{2  r}  - \eta^{-2 r})^d}{(\eta^{2  m}  - \eta^{-2 m})^d} = \frac{\eta^{2 d r}  - \eta^{-2d r}}{\eta^{2d  m}  - \eta^{-2d m}} = \gamma_r
\end{align}
(since $\eta^d = -\eta^{-1}$).  So even though it is defined in terms of the quantity $\eta$, which lies outside the field $\mathbb{F}_d$, $\gamma_r$ itself lies in the field $\mathbb{F}_d$.  
Finally, it should be observed that the matrix $G$ is not unique. In fact if $k$ is any integer relatively prime to $d^2-1$ then $\bar{\theta}^k$ is also a primitive element of $\mathbb{F}_{d^2}$.  So $\eta$ in the above formulae can be replaced with $\eta^k$, for any such $k$.  For the convenience of the reader one possible choice of the matrix $G$ is listed in Table~\ref{tble:GMatrices} for each of the first $18$ primes $=5$ (mod $6$).

\begin{table}[htb]
\begin{tabular}{|p{0.2 in}p{0.9 in}|p{0.2 in}p{0.9 in}|p{0.2 in}p{0.9 in}|}
\hline
\myStrut
 \tbins{0.2}{$d$} & \tbins{0.9}{$G$} &  \tbins{0.2}{$d$} & \tbins{0.9}{$G$} & \tbins{0.2}{$d$} & \tbins{0.9}{$G$} 
\\
\hline
 \tbins{0.2}{$5$} & \myStrutB \tbins{0.9}{$\bmt 1 & -1 \\  1 & 0 \emt$} & \tbins{0.2}{$47$} & \tbins{0.9}{$\bmt -12 & -3 \\ 3 & -15\emt$} & \tbins{0.2}{$101$} & \tbins{0.9}{$\bmt 9 & 42 \\ -42 & -50\emt$}
 \\
\hline
 \tbins{0.2}{$11$} & \myStrutB \tbins{0.9}{$\bmt -4 & 2 \\ -2 & -2\emt$} &  \tbins{0.2}{$53$} & \tbins{0.9}{$\bmt-22 & 8 \\-8 & -14\emt$} & \tbins{0.2}{$107$} & \tbins{0.9}{$\bmt-47 & 31 \\-31 & -16\emt$}
 \\
\hline
\tbins{0.2}{$17$} & \myStrutB \tbins{0.9}{$\bmt 6 & -8 \\ 8 & -2\emt$} &  \tbins{0.2}{$59$} & \tbins{0.9}{$\bmt -25 & 28 \\ -28 & 3\emt$} & \tbins{0.2}{$113$} & \tbins{0.9}{$\bmt 56 & -47 \\ 47 & 9 \emt$}
 \\
\hline
\tbins{0.2}{$23$} & \myStrutB \tbins{0.9}{$\bmt 11 & -7 \\ 7 & 4\emt$} &  \tbins{0.2}{$71$} & \tbins{0.9}{$\bmt 8 & 34 \\-34 & -29\emt$} & \tbins{0.2}{$131$} & \tbins{0.9}{$\bmt-44 & 11 \\-11 & -33\emt$}
 \\
\hline
 \tbins{0.2}{$29$} & \myStrutB \tbins{0.9}{$\bmt12 & -9 \\ 9 & 3\emt$} &  \tbins{0.2}{$83$} & \tbins{0.9}{$\bmt 18 & -25 \\ 25 & -7\emt$} & \tbins{0.2}{$137$} & \tbins{0.9}{$\bmt -58 & 20 \\-20 & -38\emt$}
 \\
\hline
\tbins{0.2}{$41$}& \myStrutB\tbins{0.9}{$\bmt7 & 11 \\-11 & 18\emt$}&  \tbins{0.2}{$89$} & \tbins{0.9}{$\bmt -25 & -12 \\ 12 & -37\emt$} & \tbins{0.2}{$149$} & \tbins{0.9}{$\bmt47 & -12 \\12 & 35\emt$}
\\
\hline
\end{tabular}
\vspace{0.15 in}
\caption{One possible choice for the matrix $G$ for each of the first $18$ prime dimensions $=5$ (mod $6$).  Note that other choices are possible.}
\label{tble:GMatrices}
\end{table}

Now consider the unitary $U_G$.  Since $U_G$ is order $d+1$ its eigenvalues are powers of $\sigma = e^{\frac{\pi i}{3m}}$.  The  projector onto the eigenspace of $U_G$ with eigenvalue $\sigma^r$ is
\be
P_r = \frac{1}{6m} \sum_{s=0}^{6m-1} \sigma^{-r s} U^s_G
\ee
To find the dimensions of the eigenspaces we need to calculate $\Tr(U^s_G)=\Tr(U_{G^s})$.  This can be done using Theorem~\ref{thm:UtermsD}.  Let $t_s = \Tr(G^s)$.    We have
\be
t_s -2 =  \gamma_{m+s}+\gamma_{m-s} = (\eta^s - \eta^{-s})^2
\ee
So $t_s = 2$ if and only if $s=0$ (mod $6m$).  Also  
\be
(\eta^s - \eta^{-s})^d = \eta^{ds} - \eta^{-ds} = (-1)^{s+1} (\eta^s - \eta^{-s})
\ee
(since $\eta^{ds} = (-1)^s \eta^{-s}$).  So if $s\neq 0$ (mod $6m$) then $\sqrt{t_s-2}\in \mathbb{F}_d$ if and only if $s$ is odd.  Consequently
\be
l(t_s-2) = (-1)^{s+1} =- \sigma^{3ms}
\ee 
for all $s\neq 0$ (mod $6m$).  We also have $\gamma_s = 0$ if and only if $s=0$ (mod $3m$).  So it follows from Theorem~\ref{thm:UtermsD} that
\be
\Tr(U^s_G) = -\sigma^{3ms}
\label{eq:G3sMat}
\ee
if $s\neq 0$ (mod $3m$).  It is easily seen that
\begin{align}
G^{3m} = \bmt -1 & 0 \\ 0 & -1\emt
\end{align}
Taking into account the standard result~\cite{Hardy} 
\be
l(-1) = (-1)^{\frac{d-1}{2}} = -(-1)^m
\ee
we deduce that Eq.~(\ref{eq:G3sMat}) continues to hold when $s= 3m$.  
Consequently
\be
\Tr(P_r)  = \frac{1}{6m}\left(d -\sum_{s=1}^{6m-1} \sigma^{(3m-r)s }\right) = 1- \delta_{r,3m}
\ee
So the eigenspaces are all one dimensional, with the exception of the eigenspace corresponding to $r=3m$ which is zero dimensional.

Let $|\psi_r\rangle$ be the normalized eigenvector corresponding to the eigenvalue $\sigma^r$ (notice that $|\psi_r\rangle$ can be calculated, up to a phase, from the explicit formula we have given for $P_r$).  We  have $U_F |\psi_r\rangle = \lambda^r |\psi_r\rangle$, where $\lambda = \sigma^{2m} = e^{\frac{2 \pi i}{3}}$.  Let $\mathcal{S}_k$ be the eigenspace of $U_F$ with eigenvalue $\lambda^k$.  Then bases and dimensions of these subspaces are as in Table~\ref{tble:zaunerBases}. Scott and Grassl's exhaustive numerical investigation~\cite{ScottGrassl} suggests that fiducial vectors always exist in $\mathcal{S}_1$ and $\mathcal{S}_2$, and that if $m=0$ (mod $3$) they also exist in $\mathcal{S}_0$.
 \vspace{0.15 in}
\begin{table}[htb]
\begin{tabular}{|p{0.8 in}|p{2.7 in}|p{0.8 in}|}
\hline
\myStrut \tbins{1}{subspace}   & \tbins{2.7}{basis} &  \tbins{0.8}{dimension}
\\
\hline
\myStrut \tbins{0.8}{$\mathcal{S}_0$} & \tbins{2.7}{$|\psi_0 \rangle,|\psi_3\rangle, \dots , |\psi_{3m-3}\rangle , |\psi_{3m+3}\rangle, \dots ,|\psi_{6m-3}\rangle$ } &\tbins{0.8}{$2m-1$}
\\
\hline
\myStrut \tbins{0.8}{$\mathcal{S}_1$} & \tbins{2.7}{$|\psi_1 \rangle,|\psi_4\rangle,  \dots ,|\psi_{6m-2}\rangle$ } &\tbins{0.8}{$2m$}
\\
\hline
\myStrut \tbins{0.8}{$\mathcal{S}_2$} & \tbins{2.7}{$|\psi_2 \rangle,|\psi_5\rangle,  \dots ,|\psi_{6m-1}\rangle$ } &\tbins{0.8}{$2m$}
\\
\hline
\end{tabular}
\vspace{0.15 in}
\caption{}
\label{tble:zaunerBases}
\end{table}

\section{Representing the $F$ Matrices as Permutation Matrices}
\label{sec:FsAsPerms}
The matrices $F\in\ESL(2,\mathbb{F})$ permute the vectors $\mathbf{u}\in(\mathbb{F}_d)^2$ which label the displacement operators.  For many purposes it would be convenient if we could label the vectors by an integer $r$, so that $F$ acts according to the rule $F\colon \mathbf{u}_r \to \mathbf{u}_{f(r)}$ for some function $f$ acting on the index (it will, for instance, be convenient in Section~\ref{sec:reLabelling2}).  Of course one could assign the integer label arbitrarily.  However, one would like to find a natural way of doing it, such that the function $f$ takes a simple form for every $F\in \ESL(2,\mathbb{F}_d)$.

We can achieve this by borrowing an idea from quantum optics.  In the analysis of continuous variable systems it is often convenient to replace the two quadratures $q$, $p$ with the single complex variable
\be
\alpha = \frac{1}{\sqrt{2}} (q + i p)
\ee
One can obtain a discrete analogue of this  by letting the extension field  $\mathbb{F}_{d^2}$ play the role of $\mathbb{C}$.  We accordingly define, for each $\mathbf{u} \in \left(\mathbb{F}_d\right)^2$,
\be
x(\mathbf{u}) = \eta^{-1} u_1 + u_2 
\label{eq:uToxDef}
\ee
(one is free to use any basis for  $\mathbb{F}_{d^2}$ in this definition; we choose the  basis $\{\eta^{-1},1\}$  because that will prove convenient in Section~\ref{sec:reLabelling2}).  The map $\mathbf{u} \to x(\mathbf{u})$ is easily seen to be a bijection of $(\mathbb{F}_d)^2$ onto $\mathbb{F}_{d^2}$, with inverse 
\be
x \to \mathbf{u}(x) = 
\bmt \frac{x-x^d}{\eta+\eta^{-1}} \\ \vphantom{\biggl(}\frac{\eta x + \eta^{-1} x^d}{\eta+\eta^{-1}}
\emt
\ee
This suggests that we define
\be
D^{\mathrm{e}}_x = D_{\mathbf{u}(x)}
\ee
(``e'' for ``extension field labelling'').  We then have 
\be
D^{\mathrm{e}}_x D^{\mathrm{e}}_{x'} = \tau^{\langle x,x'\rangle} D^{\mathrm{e}}_{x+x'}
\ee
for all $x$, $x'\in \mathbb{F}_{d^2}$, where
\be
\langle x,x'\rangle = \tr\left(\frac{x^d x'-x {x'}^d}{\eta+\eta^{-1}}\right)
\ee

Let us now consider the action of an arbitrary symplectic/anti-symplectic matrix 
\be
F=\bmt \alpha & \beta \\ \gamma & \delta \emt
\ee
It is readily verified that
\be
U^{\dagger}_F D^{\mathrm{e}}_{x} U^{\dagger}_F = D^{\mathrm{e}}_{a x+ b x^d}
\ee
for all $x$, where 
\begin{align}
a & = \frac{\eta^{-1} \alpha + \beta +  \gamma + \eta \delta}{\eta+\eta^{-1}}
\label{eq:aDef}
\\
b & = \frac{-\eta^{-1} \alpha + \eta^{-2} \beta -  \gamma + \eta^{-1}\delta}{\eta+\eta^{-1}}
\label{eq:bDef}
\end{align}
It is worth examining the map defined by Eqs.~(\ref{eq:aDef}), (\ref{eq:bDef}) in a little more detail.  In the first place observe that it is defined for all matrices $F$ (including the singular ones), and not just for the matrices $\in \ESL(2,\mathbb{F}_d)$.  Furthermore, it is invertible, with inverse
\begin{align}
\alpha & = \frac{\eta^{-1} (a - b^d) +\eta( a^d - b)}{\eta+\eta^{-1}}
\label{eq:permAdef}
\\
\beta & = \frac{(a- b^d)-(a^d -b) }{\eta+\eta^{-1}}
\label{eq:permBdef}
\\
\gamma & = \frac{(a+\eta^{-2} b^d)-(a^d +\eta^2 b) }{\eta+\eta^{-1}}
\label{eq:permCdef}
\\
\delta & = \frac{\eta( a + \eta^{-2} b^d)+\eta^{-1}( a^d +\eta^2 b) }{\eta+\eta^{-1}}
\label{eq:permDdef}
\end{align}
It is thus a bijection of the set of all $2\times 2$ matrices over $\mathbb{F}_d$ onto the set of all pairs $\in \left(\mathbb{F}_{d^2}\right)^2$. It obviously preserves the additive structure.  Moreover, if we equip $\left(\mathbb{F}_{d^2}\right)^2$ with the product rule
\be
(a,b) \circ (a',b') = (a a' + b {b'}^d, a b' + b {a'}^d)
\label{eq:abProdRule}
\ee
it also preserves the multiplicative structure (it is a ring isomorphism, in other words).
Observe, further, that if $F$ is any $2\times 2$ matrix over $\mathbb{F}_d$, and if $(a,b)$ is the corresponding element of $\left( \mathbb{F}_d\right)^2$, then
\be
\det F = a^{d+1}- b^{d+1}
\ee
So we can identify $\ESL(2,\mathbb{F}_d)$ with the set
\be
\{(a,b) \in \left(\mathbb{F}_{d^2}\right)^2 \colon a^{d+1} - b^{d+1} = \pm 1\}
\ee
equipped with the above product rule.  

We are now ready to address the problem with which we started.   For each $0\le r \le d^2-2$ define
\be
\mathbf{u}_r  = \mathbf{u}(\bar{\theta}^{-r})
\ee
(we define it like this, with a negative exponent of $\bar{\theta}$, with a view to later convenience).
We then have
\be
F \mathbf{u}_r = \mathbf{u}_{r-\log_{\bar{\theta}} (a+ b \eta^{-2r})}
\ee
So this gives us a natural way to represent $\ESL(2,\mathbb{F}_d)$ as a group of permutation matrices.  Note that the zero vector is excluded from the labelling scheme; however this does not matter because $\boldsymbol{0}$ is invariant under the action of  $\ESL(2,\mathbb{F}_d)$.

Finally observe that if
\be
A = \bmt 0 & 1 \\ 1 & \eta - \eta^{-1} \emt
\ee  
then 
\be
A \mathbf{u}_r = \mathbf{u}_{r-\frac{d-1}{2}}
\ee
So $A$ is  represented by a shift operator.  This is a consequence of 
choosing the basis $\{\eta^{-1},1\}$ in Eq.~(\ref{eq:uToxDef}). With a different choice of basis for $\mathbb{F}_{d^2}$ one can arrange that any other order $2(d+1)$ anti-symplectic becomes a shift operator.

\section{The MUB Cycling Problem:  Preliminaries}
\label{sec:MUBPrelim}
Wootters and Sussman~\cite{WoottersA} have shown that if $d$ is a power of $2$ there exists a  Clifford unitary which cycles through a full set of mutually unbiassed bases (MUBs).  In the next two sections we address the problem of generalizing their result to the case when $d$ is an odd prime power dimension.

Let us stress that, throughout the remainder of this paper, we use such phrases as ``the full set of  MUBs'' to  mean ``the full set of MUBs generated by acting on the standard basis with elements of the group $\re$'' (which, as we will see, is just the set of MUBs orignally described by Wootters and Fields~\cite{WoottersB}). It is possible to generate other MUBs by acting with the larger group $\ure$.  Also there is, of course, the possibility that there exist full sets of MUBs which are not covariant with respect to any version of the Clifford group.

We will show:
\begin{enumerate}
\item  There is no Clifford unitary which cycles through all the MUBs for any $d$.  
\item If $d=3$ (mod $4$) there is an anti-Clifford anti-unitary which cycles through all the MUBs.
\item For every $d$ the MUBs can be divided into $2$ groups of $\frac{d+1}{2}$ in such a way that there is a  Clifford unitary which cycles through all the MUBs in each group separately. 
\end{enumerate}
Let us stress once again that we are confining ourselves to the unitaries/anti-unitaries $\in \re$. As Gross and Chaturvedi~\cite{DavidSubhash} have observed the considerations which follow are not sufficient to exclude the possibility that there exist unitaries  or (when $d=1$ (mod $4$)) anti-unitaries  $\in \ure$ which cycle through all the MUBs (except, of course, in prime dimension when the restricted and unrestricted groups are identical).

In a subsequent paper  this result will be used to show that minimum uncertainty states~\cite{WoottersA,Sussman,selfA} exist in every odd prime power dimension, thereby extending Wootters and Sussman's result~\cite{WoottersA,Sussman}, that such states exist in every even prime power dimension, and in a certain infinite class of dimensions $=3$ (mod $4$).  

The purpose of this section is to fix notation, and to derive a formula which describes how an arbitrary matrix $F\in\ESL(2,\mathbb{F}_d)$ moves us between the different MUBs.  In the next section we will prove our main result.

Let
\be
F=\bmt \alpha & \beta \\ \gamma & \delta \emt
\ee
be any element of $\SL(2,\mathbb{F}_d)$.  It can be seen from Eq.~(\ref{eq:symRepDef1}) that if $\beta =0$ then $U_F$ simply permutes and re-phases the standard basis.  If, on the other hand, $\beta\neq 0$ every element of $U_F$ has modulus $\frac{1}{\sqrt{d}}$.  So if we define 
\begin{align}
H_{\mu} &= 
\begin{cases}
\bmt 1 & \mu \\ 0 & 1\emt \qquad & \mu \neq \infty
\\ \text{\myStrutB}
\bmt  0 & 1 \\ -1 & 0 \emt \qquad & \mu =\infty
\end{cases}
\\
\intertext{and}
|\mu,x\rangle & = U_{H_\mu} |x\rangle
\end{align}
and allow  $\mu$ to vary over the set $\mathbb{F}_d \cup \{\infty\}$, this will give us a full set of $d+1$ MUBs.  The reader will easily perceive that, \emph{modulo} some differences in the phase and labelling conventions, these are  the MUBs originally described  by Wootters and Fields~\cite{WoottersB}.

Let us now consider the action of the extended Clifford group operators.  For the displacement operators we have
\be
D_{\mathbf{u}} |\mu,x\rangle 
=
\begin{cases}
\tau^{\tr((2x+u_1-\mu u_2)u_2 ))}|\mu,x+u_1-\mu u_2\rangle \qquad & \mu \neq \infty
\\
\tau^{\tr((2x-u_2)u_1)}|\mu,x-u_2\rangle \qquad & \mu = \infty
\end{cases}
\label{eq:DopOnMUBs}
\ee
Turning to the case of the symplectic and anti-symplectic operators let
\be
F = \bmt \alpha & \beta \\ \gamma & \delta \emt 
\ee
be any element of $\ESL(2,\mathbb{F}_d)$, and let $\Delta = \det F$.  Then
\be
F H_{\mu}
=
\begin{cases}
\myStrutB \bmt 1 & \frac{\alpha \mu + \beta}{\gamma \mu + \delta} \\ 0 & 1 \emt \bmt \frac{\Delta}{\gamma \mu +\delta} & 0 \\ \gamma & \gamma \mu +\delta \emt 
\qquad & \text{$\mu\neq \infty$, $\gamma \mu + \delta \neq 0$}
\\ \myStrutB
\bmt 0 & 1 \\ -1 & 0 \emt \bmt -\gamma & 0 \\ \alpha & -\frac{\Delta}{\gamma} \emt
\qquad & \text{$\mu\neq \infty$, $\gamma \mu + \delta =0$}
\\ \myStrutB
\bmt 1 & \frac{\alpha}{\gamma} \\ 0 & 1\emt \bmt \frac{\Delta}{\gamma} & 0 \\ -\delta & \gamma \emt
\qquad & \text{$\mu= \infty$, $\gamma \neq 0$}
\\ \myStrutB
\bmt 0 & 1 \\ -1 & 0\emt \bmt \delta & 0 \\ - \beta & \alpha \emt
\qquad & \text{$\mu= \infty$, $\gamma = 0$}
\end{cases}
\ee
Consequently
\be
U_F |\mu, x\rangle 
=
\begin{cases}
\myStrutC
l\left(\frac{\Delta }{\gamma \mu + \delta} \right) \tau^{\tr\left( \frac{\Delta \gamma x^2 }{\gamma \mu + \delta} \right)} \left| \frac{\alpha \mu + \beta}{\gamma \mu + \delta} , \frac{\Delta x}{\gamma \mu + \delta}
\right>
\qquad & \text{$\mu\neq \infty$, $\gamma \mu + \delta \neq 0$}
\\ \myStrutC
l\left( -\gamma \right) \tau^{-\tr\left(\alpha \gamma x^2 \right)} \left| \infty, -\gamma x\right>
\qquad & \text{$\mu\neq \infty$, $\gamma \mu + \delta =0$}
\\ \myStrutC
l\left(\frac{\Delta}{\gamma}  \right) \tau^{-\tr\left(\frac{\Delta \delta x^2}{\gamma}  \right)} \left| \frac{\alpha }{\gamma}, \frac{\Delta x}{\gamma} \right>
\qquad & \text{$\mu= \infty$, $\gamma \neq 0$}
\\ \myStrutC
l\left(\delta \right) \tau^{-\tr\left(\beta \delta x^2  \right)} \left| \infty, \delta x \right>
\qquad & \text{$\mu= \infty$, $\gamma = 0$}
\end{cases}
\label{eq:FonMUBsWFLabel}
\ee
In particular the Wootters-Fields MUBs are the only MUBs which can be got by acting on the standard basis with the unitaries/anti-unitaries $\in \re$.

\section{The MUB Cycling Problem:  Main Results}
\label{sec:MUBCyclingMain}
We are now ready to prove our main results.
The matrix $H_\mu$ is order $p$ if $\mu \neq \infty$, and order $4$ if $\mu = \infty$.  So if one wants to generate all the MUBs by acting on the standard basis with powers of the operators $U_{H_\mu}$ one needs to use a minimum of $p^{n-1}+1$ different operators.  We will now show that it is possible to generate them using a much smaller number of operators.    Specifically, we will show that in every odd prime power dimension it is possible to generate every MUB by acting on two distinguished bases with powers of a single unitary operator, and that if $d=3$ (mod $4$) one can  generate every MUB by acting on any single  basis with powers of a single anti-unitary operator.  

It can be seen from Eq.~(\ref{eq:DopOnMUBs}) that the displacement operators simply permute the vectors within a basis.  So if one is interested in finding operators which cycle through the largest possible number of different bases one may confine oneself to the symplectic unitaries and anti-symplectic anti-unitaries. 

It follows from Eq.~(\ref{eq:FonMUBsWFLabel}) that, for any $F \in \ESL(2,\mathbb{F}_d)$,
\be
U_F |\mu , x\rangle = e^{i \theta_F(\mu,x)} | f_F(\mu),g_F(\mu,x)\rangle
\ee
for suitable functions $\theta_F$, $f_F$, $g_F$.  We will say that $F$ is a cycling matrix if one gets the complete set of MUBs by acting on any single MUB with powers of $U_F$.     

Let us begin by observing that if $F$ is a cycling matrix, then so is every conjugate of $F$.  In fact suppose that $F$ is cycling, and suppose $F'=S F S^{-1}$ for some $S\in\ESL(2,\mathbb{F}_d)$. The fact that $F$ is cycling means that, for all $
\mu$,  $f_F^m(\mu) = \mu$ if and only if $m=0$ (mod $d+1$).  Since $f_{F'}^m(\mu) = \mu$ if and only if $f_F^m (f_S(\mu)) = f_S(\mu)$ it follows that $f_{F'}$ has the same property, implying that $F'$ is also a cycling matrix.

It will be convenient to define the cycling index $c_F$ to the smallest positive integer $m$ such that $f_F^m (0) = 0$.   In view of Eq.~(\ref{eq:symRepDef1}) $c_F$ is the smallest integer $m$ such that $F^m$ is of the form
\be
F^m = \bmt \alpha & 0 \\ \gamma & \delta\emt
\ee

Since $c_F$ divides the order of $F$, and since the order of every type $1$ matrix is $\le d-1$, we can immediately conclude that no type $1$ matrix is a cycling matrix.  It is also easily from Tables~\ref{tble:type3OrddEq1} and~\ref{tble:type3OrddEq2} together with Eq.~(\ref{eq:type3ordA}) that $c_F = 1$ or $p$ for every type $3$ matrix.  So if cycling matrices exist they have to be type $2$.  Every type $2$ matrix is conjugate to a power of the anti-symplectic matrix
\be
A = \bmt 0 & 1 \\ 1 & \eta-\eta^{-1} \emt
\label{eq:canAdef}
\ee
It is easily seen that if $A^r$ is a cycling matrix then so is $A$.  We conclude that cycling matrices exist if and only if $A$ is a cycling matrix.  

To calculate $c_A$  we use the fact
\be
A= S \bmt \eta & 0 \\ 0 & -\eta^{-1} \emt S^{-1}
\ee
where
\be
S=\bmt 1 & \frac{1}{\eta+\eta^{-1}} \\ \eta & -\frac{\eta^{-1}}{\eta+\eta^{-1}} \emt
\ee
Consequently
\be
A^r  = S \bmt \eta^r & 0 \\ 0 & (-1)^r\eta^{-r} \emt S^{-1}
= \bmt \beta_{r-1} & \beta_r\\ \beta_r & \beta_{r+1}\emt
\ee
where 
\be
\beta_r = \frac{\eta^{dr} - \eta^r}{\eta^d - \eta}
\ee
(where we have used the fact that $\eta^d = -\eta^{-1}$).  So $c_A$ is the smallest positive integer $m$ such that
\be
\eta^{m(d-1)} = 1
\ee
Now $\eta^{m(d-1)}=1$ if and only if $m(d-1) = 0 $ (mod $2(d+1)$).  So $c_A$ is the smallest positive integer $m$ such that
\be
m(d-1) = 2 k(d+1)
\label{eq:cAcalcA}
\ee
for some $k$.  

Suppose that  $d=4u+1$ for some $u$. Then Eq.~(\ref{eq:cAcalcA}) becomes
\be
m u = k (2u+1)
\ee
Since $u$ and $2u+1$ are relatively prime we conclude that $c_A= \frac{d+1}{2}$.  So $A$ is not a cycling matrix.  

Suppose, on the other hand, that $d= 4 u +3$ (mod $4$).  Then Eq.~(\ref{eq:cAcalcA}) becomes
\be
m(2u+1) = 4k(u+1)
\ee
Since $2u+1$ and $4(u+1)$ are relatively prime we conclude $c_A=d+1$. So $A$ is a cycling matrix.

Finally, note that if $A$ is a cycling matrix then $A^r$ is also a cycling matrix if and only if $r$ and $d+1$ are relatively prime.  In particular, the fact that   $d+1$ is even means that even powers of $A$ are never cycling matrices.  So there are no symplectic cycling matrices.

We have thus proved
\begin{theorem}
Let $d$ be any odd prime power dimension.  Then
\begin{enumerate}
\item There are no symplectic cycling matrices.
\item If $d=1$ (mod $4$) there are no anti-symplectic cycling matrices either.
\item If $d=3$ (mod $4$) an anti-symplectic matrix $F$ is a cycling matrix if and only if $\Tr (F) = \eta^{r} -\eta^{-r}$ for some $r$ relatively prime to $d+1$.
\end{enumerate}
\end{theorem}

At this point let us stress once again the point we made earlier, that we are only considering the operators in $\re$.  As Gross and Chaturvedi~\cite{DavidSubhash} have observed, the above considerations do not exclude the possibility that for some or all values of $d$ there is a cycling unitary $\in \urc$; or that for some or all values of  $d=1$ (mod $4$) there is a cycling anti-unitary $\in \ure$.

We now turn our attention to what, in the absence of a cycling matrix, might be considered the ``next best thing'':  namely, a matrix $F$ with the property that the MUBs split into two groups of $\frac{d+1}{2}$, in such a way that $U_F$ cycles through each group separately.  We will refer to such matrices as half-cycling  matrices.  As with cycling matrices, if $F$ is half-cycling, then so is every conjugate of $F$.  

For given $F$ and $\mu$ define $m_F(\mu)$ to be the smallest positive integer $m$ such that $f_F^m(\mu) = \mu$.  Clearly $F$ is a half-cycling matrix if and only if $m_F(\mu) = \frac{d+1}{2}$ for all $\mu$.  

Since $m_F(\mu)$ is a factor of $\ord(F)$ for all $\mu$ no type $1$ matrix can be a half-cycling matrix.  Since $m_F(0)=1$ or $p$ for all type $3$ matrices no type $3$ matrix can be a half-cycling matrix.  So, as with cycling matrices, if half-cycling matrices exist at all they must be type $2$. It follows that every half-cycling matrix must be conjugate to a power of the matrix $A$ defined in Eq.~(\ref{eq:canAdef}).  

To show that $A^r$ is a half-cycling matrix it is enough to show that $m_{A^r}(\mu) =\frac{d+1}{2}$ for more than half the possible values of $\mu$.  So in the following we will assume that $\mu\neq \infty$.  
Observe that if
 $f^m_{A^r}(\mu) = \mu$ we must have 
\begin{align}
\left(U_{A}\right)^{mr} | \mu, x\rangle & = e^{i\theta} | \mu, k x\rangle
\\
\intertext{for some $\theta$, $k$.  Equivalently}
U_{A^{mr}_\mu} | 0, x\rangle & = e^{i\theta} | 0 , k x\rangle
\label{eq:halfCycA}
\end{align}
where
\begin{align}
A^{mr}_\mu & = H^{-1}_{\mu} A^{mr} H_{\mu} 
\nonumber
\\
& = 
\bmt \bigl(\beta_{mr-1} - \mu \beta_{mr}\bigr) & \bigl(- \beta_{mr}\mu^2 +  (\beta_{mr-1}-\beta_{mr+1} )\mu + \beta_{mr}\bigr) \\ \beta_{mr} 
& \bigl(\beta_{mr} \mu + \beta_{mr+1}\bigr)
\emt
\end{align}
Referring to Eq.~(\ref{eq:symRepDef1}) we see that this implies 
\be
\beta_{mr}\mu^2 - (\beta_{mr-1}-\beta_{mr+1} )\mu - \beta_{mr} = 0
\ee
This is a quadratic equation in $\mu$ with discriminant
\be
(\beta_{mr-1}-\beta_{mr+1} )^2 + 4 \beta_{mr}^2 = (\eta^{mr} - (-1)^{mr} \eta^{-mr})^2
\label{eq:halfCycB}
\ee
If the equation has solutions in $\mathbb{F}_d$ the discriminant must $\in Q \cup \{0\}$, implying $\eta^{mr} - (-1)^{mr} \eta^{-mr}\in\mathbb{F}_d$.  But
\be
\left( \eta^{mr} - (-1)^{mr} \eta^{-mr} \right)^d = -\left( \eta^{mr} - (-1)^{mr} \eta^{-mr} \right)
\ee
which means that $\eta^{mr} - (-1)^{mr} \eta^{-mr}\in\mathbb{F}_d$ if and only if $\eta^{mr} - (-1)^{mr}\eta^{-mr}=0$. So $f^{m}_{A^r} (\mu) = \mu$ implies 
\begin{align}
\eta^{2 mr} & = (-1)^{mr}
\\
\intertext{or}
mr(d-1) &= 0 \quad \text{ (mod $2(d+1)$)}
\label{eq:halfCycC}
\end{align}

Suppose, now, that $d=1$ (mod $4$).  Then Eq.~(\ref{eq:halfCycC}) reads
\be
mr\left(\frac{d-1}{4}\right) = 0 \quad \left(\text{mod $\frac{d+1}{2}$}\right)
\ee
Since $\frac{d-1}{4}$ is relatively prime to $\frac{d+1}{2}$ this implies
\be
mr = 0 \quad \left(\text{mod $\frac{d+1}{2}$}\right)
\ee
We have thus shown that a necessary condition for $f^{m}_{A^r}(\mu) = \mu$ is that $mr$ is a multiple of $\frac{d+1}{2}$. But it is readily verified that 
\be
A^{\frac{k(d+1)}{2}}_\mu = \bmt \eta^{\frac{k(d+1)}{2}} & 0 \\ 0 & \eta^\frac{k(d+1)}{2} \emt
\ee
for all $\mu$.  So the condition is not only necessary but also sufficient. It follows that, for all $\mu\neq \infty$,
\be
m_{A^r}(\mu) = \frac{\frac{d+1}{2}}{\left[ r, \frac{d+1}{2}\right]}
\ee
 In particular, $A^r$ is a half-cycling matrix if and only if $r$ is relatively prime to $\frac{d+1}{2}$. 

Suppose, on the other hand, that $d=3$ (mod $4$). Then Eq.~(\ref{eq:halfCycC}) reads
\be
mr\left(\frac{d-1}{2}\right) = 0 \quad \left( \text{mod $d+1$} \right)
\ee
So a necessary condition for $f^{m}_{A^r}(\mu) = \mu$ is that $mr$ is a multiple of $d+1$. But it is readily verified that 
\be
A^{k(d+1)}_\mu = \bmt (-1)^k & 0 \\ 0 & (-1)^k \emt
\ee
irrespective of the value of $\mu$.  So the condition is also sufficient.  It follows that, for all $\mu\neq \infty$,
\be
m_{A^r}(\mu) = \frac{d+1}{\left[ r, d+1\right]}
\ee
In particular, $A^r$ is a half-cycling matrix if and only if $r$ is an even integer such that $\frac{r}{2}$ is relatively prime to $\frac{d+1}{2}$.  

We have thus proved
\begin{theorem}
Let $d$ be any odd prime power dimension.  Then
\begin{enumerate}
\item[(a)] A symplectic matrix $F$ is a half-cycling matrix if and only if $\Tr(F)=\eta^{2r}+\eta^{-2r}$ for some integer $r$ relatively prime to $\frac{d+1}{2}$.
\item[(b)] If $d=1$ (mod $4$) an anti-symplectic matrix $A$ is half-cycling if and only if $\Tr(F) = \eta^{r}-\eta^{-r}$ for some odd integer $r$ relatively prime to $d+1$.  
\item[(c)] If $d=3$ (mod $4$) there are no half-cycling antisymplectic matrices.
\end{enumerate}
\end{theorem}
 
\section{Alternative Labelling Scheme}
\label{sec:reLabelling2}
If we label the MUBs in the way we have been doing so far, in terms of the parameter $\mu$, the cycling structure described in the last section is not very obvious.  We would like to construct a labelling scheme which makes it more explicit.  

This is, in essence, an easy problem to solve:  all we have to do is choose an order $2(d+1)$ anti-symplectic  $A$ and then  label an arbitrary MUB by the number of applications of $U_A$ required to reach it, starting from some distinguished MUB on the same orbit. If $d=3$ (mod $4$) it becomes easier still:  all one has to do is  count the number of applications of $U_A$ needed to reach the MUB of interest starting from the standard basis.  However, there is a slight complication due to the connection between the MUBs and the displacement operators:  each MUB determines a maximal commuting set of displacement operators and conversely (namely, the set of displacement operators which  it diagonalizes).  We would like our labelling scheme to make this relationship transparent.  To do that we will make use of the labelling of the displacement operators described in Section~\ref{sec:FsAsPerms}.

It is easily seen that the maximal commuting sets of displacement operators corresponding to the  Wootters-Fields MUBs are all of the form
\be
\{D_{w \mathbf{u}} \colon w \in \mathbb{F}_d \}
\ee
for some fixed $\mathbf{u}\neq \boldsymbol{0}$.  
In terms of the labelling of the displacement operators described in Section~\ref{sec:FsAsPerms} we have
\be
w \mathbf{u}_r = \mathbf{u}_{r-(d+1) \log_{\theta} w}
\ee
for all non-zero $w\in \mathbb{F}_d$.  This suggests replacing the single  integer $r$ with a pair of indices $s,t$, defined so that $s$ (respectively $t$) is the unique integer in the range $0\le s \le d$  (respectively $0 \le t \le d-2$) such that
\be
r = s+ t (d+1)
\ee
(in other words $t$ and $s$ are, respectively, the quotient and remainder of $r$ on division by $d+1$).  In terms of this double index notation we have
\be
w\mathbf{u}_{s,t} = \mathbf{u}_{s,t-\log_\theta w}
\ee
for all non-zero $w\in \mathbb{F}_d$.  So the maximal sets of commuting displacement operators corresponding to the Wootters-Fields MUBs are the $d+1$ sets 
\be
\mathcal{S}_s= \{ D_{\mathbf{u}_{s,t}}\colon 0\le t \le d-2\}\cup\{ D_{\boldsymbol{0}}\}
\ee
obtained as $s$ ranges over the interval $0\le s \le d$.  Thus, the index $s$ labels the maximal set and the index $t$ labels the individual displacement operators within the set.  

We now define, corresponding to each  $\mathcal{S}_s$, the $d$ operators
\be
P_{s,x} = \frac{1}{d}\left( 1 + \sum_{t=0}^{d-2} \omega^{-\tr(\theta^{-t} x)} D_{\mathbf{u}_{s,t}}
\right)
\label{eq:MUBProjCycDef}
\ee
It is readily confirmed that the $P_{s,x}$ are rank $1$ projection operators and, furthermore, that
\be
\tr\left(P_{s,x} P_{s',x'}
\right)
=
\begin{cases}
\delta_{x, x'} \qquad & s = s'
\\
\frac{1}{d} \qquad & s\neq s'
\end{cases}
\ee
So if we set
\be
P_{s,x} = | s, x\rangle_{\mathrm{c}} \ \mathstrut_{\mathrm{c}}\langle s,x |
\ee
(``c'' for ``cycling'') the vectors $|s,x\rangle_{\mathrm{c}}$ will constitute the full set of MUBs, with $s$ labelling the basis and $x$ labelling the individual vector within the basis.  Moreover, $\{|s,x\rangle\colon x\in\mathbb{F}_d\}$ is the basis which diagonalizes the displacement operators in $\mathcal{S}_s$.

To find the relation between this labelling scheme and the one used previously observe that, in  view of Eq.~(\ref{eq:DopOnMUBs}),
\begin{align}
\langle \mu, y | P_{s,x} | \mu, y\rangle
& = 
\frac{1}{d} \left( 1 + \sum_{t=0}^{d-2} \omega^{-\tr(\theta^{-t} x)} \langle \mu,y| D_{\mathbf{u}_{s,t}} |\mu,y\rangle \right)
\nonumber
\\
& = 
\begin{cases}
\frac{1}{d} \qquad & \mu \neq \mu_s 
\\
\delta_{y,\lambda_s x} \qquad & \mu = \mu_s
\end{cases}
\intertext{implying}
|s,x\rangle_{\mathrm{c}} &\dot{=} |\mu_s, \lambda_s x \rangle
\end{align}
where the notation ``$\dot{=}$'' means ``equals up to a phase'', and where
\begin{align}
\mu_s & = 
\begin{cases}
\frac{\eta^s -\eta^{-s}}{\eta^{s+1} + \eta^{-s-1}}  \qquad & s\neq \frac{d-1}{2}
\\
\infty \qquad & s = \frac{d-1}{2}
\end{cases}
\\
\lambda_s 
& =
\begin{cases}
\frac{\eta+\eta^{-1}}{\eta \bar{\theta}^{-s} + \eta^{-1} \bar{\theta}^{-ds}}
\qquad & s\neq \frac{d-1}{2}
\\
1 \qquad & s= \frac{d-1}{2}
\end{cases}
\end{align}
We also have the inverse relation
\begin{align}
|\mu, x\rangle &\dot{=} |s_\mu, \frac{1}{\lambda_s} x\rangle_{\mathrm{c}}
\intertext{where}
s_{\mu} &= 
\begin{cases}
\frac{1}{2}\log_{\eta} \left( \frac{1+\eta^{-1} \mu}{1-\eta \mu}  \right) \qquad & \mu \neq \infty
\\
\frac{d-1}{2} \qquad & \mu = \infty
\end{cases}
\end{align}
In particular $|0,x\rangle_{\mathrm{c}}$ is the standard basis.

To calculate the action of $U_F$, for arbitrary $F\in\ESL(2,\mathbb{F}_d)$, let $f_F(s)$ be the unique integer in the range $0 \le f_F(s) \le d$ and $g_F(s)$ the unique non-zero element of $\mathbb{F}_d$ such that
\be
s-\log_{\bar{\theta}}(a_F-b_F \eta^{-2s}) = f_F(s) + (d+1) \log_{\theta}\left( g_F(s)\right)
\ee
(so $\log_{\theta}\left(g_F(s)\right)$ is the quotient and $f_F(s)$ is the remainder on division by $d+1$).  Then
\be
F \mathbf{p}_{s,t} = \mathbf{u}_{f_F(s),t+\log_{\theta} \left(g_F(s)\right)}
\ee
In view of Eq.~(\ref{eq:MUBProjCycDef}) this means
\be
U_F P_{s,x}
=
\frac{1}{d} \left(
\sum_{t=0}^{d-2} \omega^{-\tr(\theta^{-t} g_F(s) x) } D_{\mathbf{u}_{f_F(s),t}}
\right)
\ee
implying
\be
U_F | s, x \rangle_{\mathrm{c}} \dot{=} |f_F(s),g_F(s) x\rangle_{\mathrm{c}}
\ee

Finally, let us consider the action of the operators $U_{A}$, $U_{A^2}$.   We have $a_{A} = \eta$, $b_{A}=0$.  In view of Eq.~(\ref{eq:abProdRule}) this means $a_{A^m} = \eta^m$, $b_{A^m}=0$ for all $m$.  So
\be
U_{A^2}|s,x\rangle_{\mathrm{c}} \dot{=} 
\begin{cases}
|s+2,\theta^{-1} x\rangle_{\mathrm{c}} \qquad & s< d-1
\\
|s+2,x\rangle_{\mathrm{c}} \qquad &\text{$s=d-1$ or $d$}
\end{cases}
\ee
(where addition in the first argument is mod $d+1$).  We see that, for all $d$, $U_{A^2}$ cycles through the even-index and odd-index MUBs separately.

Turning to $U_A$ we have
\be
U_A|s,x\rangle_{\mathrm{c}}
\dot{=}
\begin{cases}
\left|s+\frac{d+3}{2}, \theta^{-1} x  \right>_{\mathrm{c}} \qquad & s<\frac{d-1}{2}
\\
\left|s+\frac{d+3}{2}, x  \right>_{\mathrm{c}} \qquad & s\geq\frac{d-1}{2}
\end{cases}
\ee
If $d=1$ (mod $4$) then $\frac{d+3}{2}$ is even, so $U_{A}$, like $U_{A^2}$,  cycles through the even-index and odd-index MUBs separately.  But if $d=3$ (mod $4$) then $\frac{d+3}{2}$ is an odd integer relatively prime to $d+1$, so $U_{A}$ cycles through all the MUBs.  

Lastly, let us note that, although $A$ is a cycling matrix when $d=3$ (mod $4$), it is rather more convenient to use  $A^{\frac{d+3}{2}}$  as this increases the index by $1$ each time:
\be
U_{A}^{\frac{d+3}{2}} |s,x\rangle_{\mathrm{c}}
\dot{=}
\begin{cases}
|s+1,\theta^{-\frac{d+1}{4}}x \rangle_{\mathrm{c}} \qquad & s<d
\\
|s+1,\theta^{-\frac{d-3}{4}}x \rangle_{\mathrm{c}}\qquad & s= d
\end{cases}
\ee
\section{Concluding Remark}
In this paper we have confined ourselves to the group $\re$.  It would obviously be interesting to see how far our results extend to the larger group $\ure$.  It would also be interesting to examine the case of even prime power dimension.  We hope to address these questions in future publications.

\subsubsection*{Acknowledgements}
We are  grateful to Ingemar Bengtsson, Subhash Chaturvedi, Markus Grassl, David Gross, Andrew Scott and Bill Wootters for many stimulating discussions.
  
This research was supported in part by the U. S. Office
of Naval Research (Grant No.\ N00014-09-1-0247). Research
at Perimeter Institute is supported by the Government
of Canada through Industry Canada and by the
Province of Ontario through the Ministry of Research \&
Innovation.

}

\end{document}